\documentclass[%
 reprint, 
 amsmath,amssymb,
 aps, physrev,
]{revtex4-2}

\usepackage[utf8]{inputenc}
\usepackage[T1]{fontenc}
\usepackage{graphicx}
\usepackage{dcolumn}
\usepackage{bm}

\usepackage{booktabs}
\usepackage{amsfonts}
\usepackage{nicefrac}
\usepackage{microtype}
\usepackage{amsthm} 
\usepackage{mathptmx}

\usepackage{braket}   
\usepackage{mathtools}
\usepackage{tabularx}
\usepackage{threeparttable}
\usepackage{makecell}
\usepackage{mhchem}

\usepackage{hyperref}
\hypersetup{
    pdftitle={Scalable O(1) Entanglement Routing},
    pdfauthor={Xiaoyi Zheng, Chan-Tong Lam},
}

\newtheorem{lemma}{Lemma}
\newtheorem{theorem}{Theorem}

\let\bell\relax 
\newcommand{\bell}{\ensuremath{\ket{\Phi^{+}}}}         
          
\newcommand{\cz}{\ensuremath{\mathrm{CZ}}}             
\newcommand{\cluster}{\ensuremath{\mathrm{LC}}}         
             
\newcommand{\CZ}{\text{CZ}}
\newcommand{\LC}{\text{LC}}

\newcounter{protocol}

\begin{document}

\title{Scalable Graph State Generation with $O(1)$ Local Feedforward in Quantum Networks}

\author{Xiaoyi Zheng}
\email{xiaoyi.zheng@mpu.edu.mo}
\affiliation{%
 Faculty of Applied Sciences, Macao Polytechnic University, 999078, Macao SAR, China
}%

\author{Chan-Tong Lam}
\email{ctlam@mpu.edu.mo}
\affiliation{%
 Faculty of Applied Sciences, Macao Polytechnic University, 999078, Macao SAR, China
}%

\author{Lin Chen}
\email{lchen@mpu.edu.mo}
\affiliation{%
 Faculty of Applied Sciences, Macao Polytechnic University, 999078, Macao SAR, China
}%

\author{Zheng Xing}
\email{201135@gwng.edu.cn}
\affiliation{%
 School of Computer Science, South China Business College, Guangdong University of Foreign Studies, Guangzhou 510545, Guangdong, China
}%

\date{\today}

\begin{abstract}
The development of quantum networks faces a key challenge: the contradiction between probabilistic long-range entanglement generation and finite coherence time. Existing routing protocols typically focus on global state computation or path optimization. As the network scales up, classical delays accumulate and exacerbate decoherence, leading to a decrease in entanglement fidelity. To reduce routing decision delays to levels far below the coherence time of qubits, we propose a protocol based on local measurement and classical feedforward. This protocol reduces the local decision complexity to amortized $O(1)$, ensuring that the decision delay is always much smaller than the coherence time of qubits. We map this protocol onto a dual-species trapped-ion platform and perform hybrid simulations. The results show that the proposed protocol performs well in terms of both resource efficiency and time feasibility. Noise analysis indicates that readout fidelity is the main bottleneck of this protocol, but noise suppression can be achieved by employing an erasure transformation in the dual-species architecture, combined with spatial multiplexing and branch independence, thereby ensuring the generation of high-fidelity star subgraphs. This protocol provides a clear path to achieving high-fidelity star subgraphs. These subgraphs can serve as general modules, merging to construct arbitrary subgraphs, providing a feasible solution for future fault-tolerant distributed quantum computing.
\end{abstract}


\maketitle

\section{Introduction}
\label{sec:Introduction}
The generation of distributed multi-party entangled states is a crucial research area in quantum networks and distributed quantum computing. Although experimental progress has been made in preparing Bell states between remote nodes~\cite{ref1,ref2,ref3,ref4,ref5,ref6,ref7,ref8,ref10,ref11,ref12}, many challenges remain in fusing these fundamental entanglements into complex multi-party graph states. One such challenge is the discrepancy between the probabilistic success of entanglement distribution and the finite coherence time of qubits. This timescale mismatch requires established links to be used within a time much shorter than the coherence time of qubits to minimize fidelity degradation caused by decoherence.

In terms of physical implementation, Trapped-ion systems, due to there long coherence times~\cite{ref16,ref17} and high-fidelity quantum logic gates and quantum measurements~\cite{ref13,ref14,ref15}, hold promise as a feasible platform for future high-performance quantum networks. Related experimental progress has also demonstrated this characteristic: for example, barium ion-based architectures have achieved local gate fidelity exceeding 99.9\%~\cite{ref_helios}; recently, a team from Tsinghua University~\cite{ref_tsinghua} also achieved long-range ion-ion entanglement with $\approx$96\% fidelity over a 1.2-kilometer optical fiber.

However, current routing protocols can be mainly divided into two paradigms: centralized protocols~\cite{refcentral1,refcentral2,refcentral3} and distributed protocols. The former is limited by the star topology, easily leading to single-point performance bottlenecks and making it difficult to support large-scale long-distance networks; the latter usually assumes that network nodes have functional homogeneity and is highly dependent on global state synchronization~\cite{refMMG,refCKL} or complex pathfinding algorithms~\cite{refFan,refevan}. However, as the network scale expands, the cumulative delay of classical signal transmission and path calculation will significantly exacerbate the decoherence effect of quantum states, thus causing the entanglement fidelity to decrease significantly with increasing scale.

To overcome the limitations of latency on the scalability of quantum networks, this paper introduces a heterogeneous network architecture from a practical perspective. By distinguishing between user nodes and relay nodes, it breaks the assumption of homogeneous node functions in existing distributed protocols and simplifies the link interconnection model. Based on this, this paper proposes a novel protocol that does not require global state synchronization or path search, relying only on local measurement and classical feedforward mechanisms. This protocol maximizes the compression of classical communication and computational latency, reducing control complexity to an amortized $O(1)$ level, ensuring that the decision delay is much smaller than the coherence time of a qubit, thereby guaranteeing the fidelity of multi-qubit entangled states.

The protocol design originates from the Pauli measurement theory of graph states~\cite{refgraph2,refgraph3}. Our research shows that when the neighborhood of the qubit to be measured constitutes an independent set, the secondary operators caused by the Pauli-X measurement (which typically exhibit complex nonlocality in general graph states) can be systematically simplified. The absence of edges within the neighborhood allows the originally cumbersome multi-qubit Pauli framework correction to collapse into deterministic local single-qubit unitary corrections. We design a protocol around this characteristic so that the graph state recovery operation does not rely on global topological metadata or perform multi-body joint operations, thereby reducing the network decision complexity to a local $O(1)$ order of magnitude. For the relevant formal proofs, see (Lemma ~\ref{lem:star-formation}, Lemma ~\ref{lem:star-phase} and Theorem ~\ref{thm:x-simplify}).

To evaluate the performance of our protocol on a physical platform, we map it onto a dual-species trapped ion platform. We define a hardware-aware architecture based on a dual-species paradigm to physically isolate routing operations from data storage~\cite{refdual1,refdual2,refdual3,refdual4}. Although we assume the network is physically isomorphic, with all nodes equipped with communication qubits (e.g., $^{138}\mathrm{Ba}^{+}$) and storage qubits (e.g., $^{171}\mathrm{Yb}^{+}$), we perform functional separation at the network layer. Specifically, relay nodes use the storage qubits as asynchronous buffers to synchronize the successful distribution of entangled qubits with randomness; while terminal nodes focus on maintaining the final entangled state. This architecture effectively shields the inevitable scattering noise and reset delays during routing within the relay layer.

We evaluated the performance of this protocol using a hybrid simulation method based on standard ion trap physics parameters. The results show that the scheme exhibits excellent resource scalability: the entanglement overhead for a single user converges to the theoretical lower bound $1+1/n$. Simultaneously, the amortized control complexity of $O(1)$ ensures that the decision delay is consistently two orders of magnitude lower than the qubit coherence time limit. For the main scalability bottleneck of readout fidelity, we propose and simulate a collaborative scheme: utilizing the erasure conversion technique of a dual-species trapped ion platform to convert bit-flipping errors into heralded erasures, and applying spatial multiplexing based on the branch independence of star-shaped subgraphs. This strategy successfully guarantees the generation of high-fidelity star-shaped subgraphs. Given that star-shaped subgraphs can construct arbitrary graphs through fusion operations, this architecture provides a clear and feasible technical path for realizing large-scale, fault-tolerant distributed quantum computing.

In summary, the main contributions of this paper can be summarized as follows:
\begin{itemize}
    \item \textbf{Theoretical Operator Reduction:} We theoretically prove that when the neighborhood of a measured qubit constitutes an independent set, the complex multi-qubit byproduct operators caused by Pauli-X measurements systematically collapse into deterministic local single-qubit unitary corrections.
    \item \textbf{Low-Latency $O(1)$ Protocol:} We propose a routing protocol based on local measurement and classical feedforward. By avoiding global state computation, it reduces the local decision complexity to an amortized $O(1)$, ensuring that routing decision delays remain far below the coherence time of qubits.
    \item \textbf{Hardware Mapping and Hybrid Simulations:} We map this protocol onto a dual-species trapped-ion platform. Simulations verify the protocol's resource efficiency and time feasibility, and demonstrate that combining erasure transformation with spatial multiplexing effectively overcomes readout fidelity bottlenecks. This ensures the reliable generation of high-fidelity star subgraphs, which serve as general modules to construct arbitrary topologies.
\end{itemize}

The remaining of the papers is organized as follows. Section~\ref{sec:protocol_design} details the protocol design, including the hardware mapping, theoretical foundations, and the core sub-protocols for subgraph generation and fusion. Section~\ref{sec:simulation} presents the hybrid simulation results, analyzing resource utilization, temporal feasibility, and noise suppression. Section~\ref{sec:comparison} provides a comparison with related works. Finally, Section~\ref{sec:conclusion} concludes the paper.

\begin{figure*}[t] 
\centering
\includegraphics[width=0.85\textwidth]{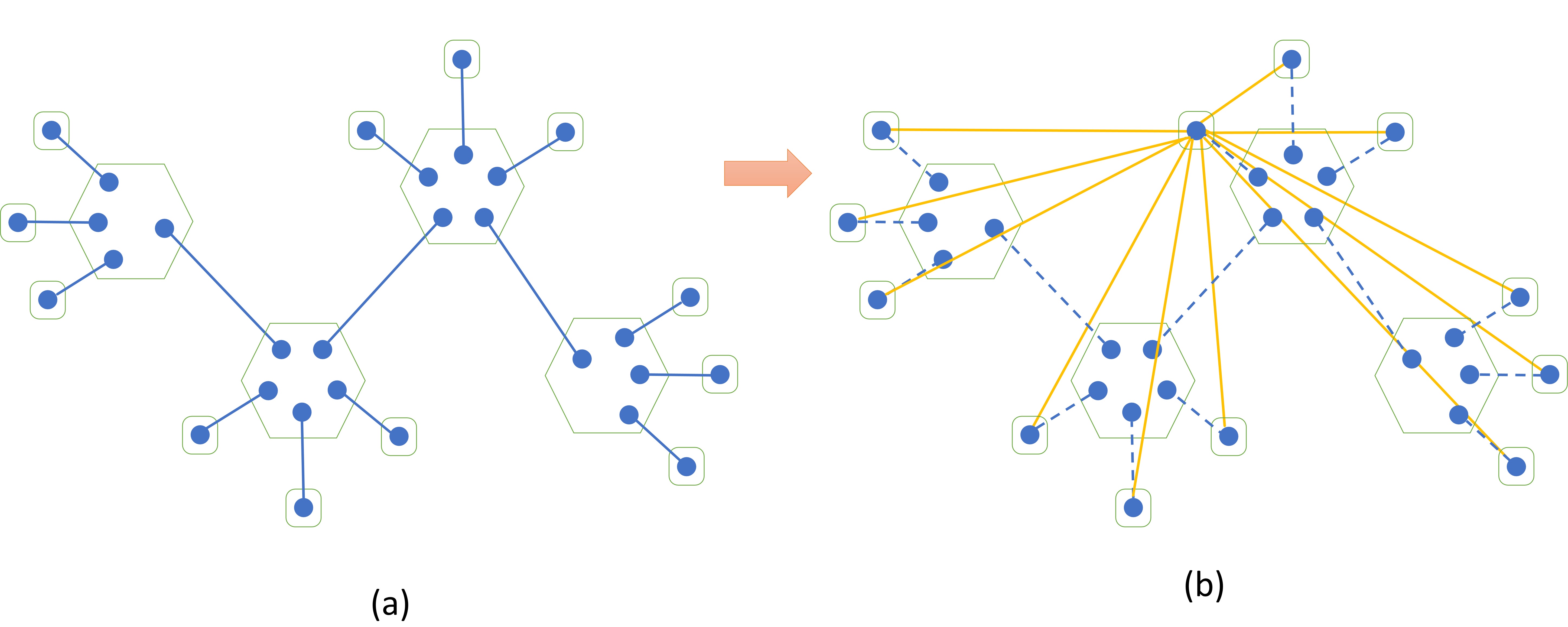}
\caption{%
\textbf{Protocol Architecture and Topology Evolution.} %
(a) \textbf{Initial Resource State.} The blue lines represent Bell states. Relay nodes (hexagons) are linearly connected via Bell states, and each terminal node (square) shares a Bell state with its nearest relay node. %
(b) \textbf{Target Graph State.} The final star-shaped graph (yellow edges represent CZ connections). The topology center is arbitrarily specified, and the final graph only contains terminal nodes; all relay nodes are no longer present in the final graph after measurement. %
} 
\label{fig1}
\end{figure*}

\section{Protocol Design}
\label{sec:protocol_design}

We present a decentralized framework for scalable star-graph generation based on functional node segregation (Fig.~\ref{fig1}a). 
terminal nodes (squares) host user qubits, performing only basic single-qubit operations (Hadamard gates and Pauli-Z gates); conversely, relay nodes (hexagons) constitute the routing infrastructure. performing the controlled Z (CZ) gates and Pauli-X basis measurements. Upon protocol completion, these relay nodes are systematically decoupled from the final state.The network is initialized as follows: each terminal node shares a pre-established Bell pair with its nearest relay node, and relay nodes share pre-established Bell pairs in the form of linear connections. The goal of this protocol is to transform this initial state into a target star graph state consisting only of terminal nodes (Fig.~\ref{fig1}b).

\subsection{Physical Implementation and Hardware Mapping}
\label{subsec:hardware}

To effectively manage the probabilistic nature of entanglement distribution, we map the protocol onto a two-species trapped ion platform. We assume the network is physically homogeneous, meaning each node (including relay and terminal nodes) is equipped with a standardized two-species system, typically using $^{138}\mathrm{Ba}^{+}$ for photon communication and $^{171}\mathrm{Yb}^{+}$ for storage~\cite{refdual1,refdual2}. Although each node is physically homogeneous, we implement functional asymmetry:

\begin{enumerate}
    \item \textbf{Relay Nodes:} These nodes utilize their $^{171}\mathrm{Yb}^{+}$ storage qubits as asynchronous buffers. Since remote entanglement generation is probabilistic and heralded, link establishment events between different ports do not occur simultaneously. The relay node immediately switches the successfully established link in the communication ion ($^{138}\mathrm{Ba}^{+}$) to the storage ion ($^{171}\mathrm{Yb}^{+}$). This buffering mechanism serves two purposes: first, it protects qubits in an established entangled link from scattering noise interference during the establishment of other entangled links; second, it allows multiple independent entangled links to be established simultaneously at different times.
    
    \item \textbf{Terminal Nodes:} These nodes use their $^{171}\mathrm{Yb}^{+}$ qubits to store the final quantum state, and ultimately only need to perform local unitary correction based on the results broadcast by the relay nodes.
\end{enumerate}

Furthermore, we can utilize spatial multiplexing techniques to increase the probability of probabilistic long-range entanglement generation. That is, a relay node can allocate k pairs of communication qubits to a single topological neighbor, thereby increasing the link establishment probability to $1-(1-p)^k$. After the first link is successfully established, subsequent redundant links can continue to be used for entanglement purification. This process ensures that the initial entangled resources can be transformed into a high-fidelity normalized state, as shown in Fig.~\ref{fig1}a.

\subsection{Theoretical Foundation: Operator Reduction}
\label{subsec:theory}
In graph measurement theory~\cite{refgraph1,refgraph2,refgraph3,refgraph4}, when we perform a Pauli-$X$ measurement on a qubit, recovering the canonical form of the remaining graph often requires specific local Clifford transformations on its neighboring qubits. As measurements are performed on different qubits, the resulting multi-qubit byproduct operators accumulate, forming complex logical dependencies. In other protocols, a technique called a Pauli frame is often used to handle this. This involves tracking operator evolution in real time and then publishing it to all nodes in the network via a classical channel. The corresponding qubit only needs to perform a one-time recovery operation before being measured. This is a mathematically ingenious approach, but in a distributed network environment, it introduces severe global feedforward dependencies into the recovery logic. Communication latency increases linearly or polynomially with the network diameter and the number of measuring qubits. This long-range dependency causes qubits to become distorted due to decoherence during the long waiting period, thus affecting the fidelity of quantum states in large-scale quantum networks.

To avoid the aforementioned difficulties, this protocol, based on the neighborhood-independent set condition, controls the communication latency to a localized constant level of $O(1)$ latency. We prove that when the neighboring qubits of the measured qubit are unconnected, i.e., forming a neighborhood-independent set, the multi-qubit byproduct operators generated by the measurement will exhibit a simplified form. This allows the multi-qubit recovery logic, which originally relied on a global Pauli frame, to be simplified into a local single-qubit recovery operation. Therefore, the recovery operation no longer needs to rely on global information or global communication, effectively avoiding long-range dependencies. Below, we describe in detail the lemmas and theorems required to give this conclusion; detailed proofs can be found in ~\ref{app:graph state}.

\begin{lemma}[Local Star Topology Formation]
\label{lem:star-formation}
Given a graph $G = (V, E)$ where the neighborhood $\mathcal{N}(i)$ of vertex $i$ is an independent set. When qubit $i$ undergoes a Pauli-$X$ measurement, the graph state evolves according to the following transformation sequence. This sequence captures the structural changes at the graph representation. For any selected neighbor $j \in \mathcal{N}(i)$, the graph topology updates as follows:

\begin{align}
G_1 &= \mathrm{LC}_j(G), \nonumber\\
G_2 &= \mathrm{LC}_i(G_1), \nonumber\\
G_2^{-i} &= G_2 \setminus \{i\}, \nonumber\\
G'' &= \mathrm{LC}_j(G_2^{-i}).
\end{align}

This transformation sequence generates a localized star topology in the final graph $G''$. This star graph is derived from the vertex set $\{j\} \cup S$, where the set $S = \mathcal{N}(i) \setminus \{j\}$. In this structure, $j$ is the node center, and $S$ is its surrounding leaf node.
Note that the rest of the graph $G''$ retains its original arbitrarily complex topology.
\end{lemma}

The neighborhood independent set condition is the foundation of our protocol design, which simplifies the algebraic form of the graph state recovery operation.

\begin{lemma}[Operator Reduction in Star Topology]
\label{lem:star-phase}
For a star graph state $|G\rangle$ with center $c$ and leaves $L$ (where $|L| \geq 1$),  the following operator equivalence exists:
\begin{equation}
Z_c H_c \bigotimes_{k \in L} Z_k |G\rangle = H_c |G\rangle
\end{equation}
\end{lemma}
This lemma shows that a collective $Z$ operation on all leaf nodes and a $ZH$ operation on the center node are equivalent to a single Hadamard operation on the center node. This symmetry allows nonlocal byproduct operators, which are originally scattered throughout the network and require multi-party collaborative processing, to be "compressed" and collapsed into local single-bit operations on the center node. Based on this property, we derive a general simplification theorem for recovery logic.

\begin{theorem}[Simplification Theorem for $|+x\rangle$ Measurement]
\label{thm:x-simplify}
When $\mathcal{N}_i(G)$ constitutes an independent set and the $X$-basis measurement at vertex $i$ yields $|+x\rangle$, the recovery operation simplifies to:
\begin{equation}
(U_+^{(i,j)})^{-1} = H_j.
\end{equation}
\end{theorem}
Theorem ~\ref{thm:x-simplify} forms the core mechanism of our protocol. We prove that, given an independent neighborhood, when the measurement result $M_i=0$, the corresponding recovery operation can be simplified to a Hadamard gate $H_j$. When $M_i=1$, a suitable reference node selection strategy can simplify the recovery operation to $Z_j H_j$. This means that the multi-particle correction, which originally grew with the neighborhood size, is simplified to a single-particle local operation targeting the reference node $j$. This property avoids the chain propagation of the repair operation in the network, thereby minimizing the fidelity degradation caused by the accumulation of classical communication delays in large-scale networks.

\subsection{Protocol~\ref{protocol:star_gen}: Star-Shaped Subgraph Generation Protocol}
\label{subsec:protocol1}
The Protocol~\ref{protocol:star_gen} provides detailed implementation steps. Relay nodes first utilize pre-shared Bell states to construct initial entanglement through local Hadamard and $CZ$ gates. Subsequently, $Pauli-X$ measurements are performed on specific qubits, and the results are broadcast.
Due to the satisfaction of the neighborhood independence condition, each terminal node only needs to respond to a 1-bit classical signal. By performing $O(1)$ local unitary corrections ($H$ or $ZH$ combinations), the nodes collaboratively generate a star subgraph.
This protocol supports parallel execution on all relay nodes. The final set of star subgraphs $\{G_{\text{sub}}^i\}$ constitutes the basic unit for subsequent topology construction.
Figure ~\ref{fig2}(a) illustrates the specific timing flow of protocol~\ref{protocol:star_gen}. Formal proofs are provided in Appendix~\ref{app:protocol1}.

\begin{figure*}[t]
\hrule height 1pt
\vspace{0.5em}
\noindent \textbf{Protocol 1: Star-Shaped Subgraph Generation Protocol} \refstepcounter{protocol}\label{protocol:star_gen}
\vspace{0.5em}
\hrule height 0.5pt
\vspace{0.5em}

\textbf{Input}: 
\begin{itemize}
    \item Relay node $\mathcal{R}$ shares Bell states with terminal nodes.
    \item Initial quantum resource: $\bigotimes_{X\in\mathcal{P}} \bell_{X_1 X_2} \otimes \bell_{T_1 T_2}$ .\\
    where: $\mathcal{P} = \{A,B,C,\ldots\}$ denotes peripheral nodes (does not include $T_1$/$T_2$). \\
    For each $X \in \mathcal{P}$: $X_1$ is the qubit at node $X$, $X_2$ at relay node $\mathcal{R}$. \\
    $T_1$ and $T_2$ are a selected Bell pair, with $T_2$ at $\mathcal{R}$ designated as the control for CZ operations.
\end{itemize}

\textbf{Output}: Star-shaped graph state centered at $T_2$: 
$|G\rangle = \left( \prod_{X\in N(T_2)} \cz_{T_2 X} \right) |+\rangle^{\otimes |N(T_2)| + 1}$ , where neighborhood $N(T_2) = \{T_1\} \cup \{X_1 | X\in\mathcal{P}\}$

\begin{enumerate}
    \item Relay node $\mathcal{R}$ applies Hadamard gates to the Bell qubits it holds to convert them into linear cluster: 
    $\bigotimes_{X\in\mathcal{P}} |\cluster\rangle_{X_1 X_2} \otimes |\cluster\rangle_{T_1 T_2}$, where $|\cluster\rangle = (H \otimes I)|\bell\rangle$
    
    \item Relay node $\mathcal{R}$ applies controlled-Z between $T_2$ and all $X_2$ qubit:  \\
    $\ket{\Psi} = \left( \prod_{X\in\mathcal{P}} \cz_{T_2 X_2} \right) \bigotimes_{X\in\mathcal{P}} |\cluster\rangle_{X_1 X_2} \otimes |\cluster\rangle_{T_1 T_2}$

    \item Relay node $\mathcal{R}$ measures all $\{X_2 | X\in\mathcal{P}\}$ in X-basis. The measurement outcome $M_X \in \{0,1\}$ determines the collapsed state: \\
    $\ket{\psi_{X_2}} = \begin{cases} \ket{+} & M_X=0 \\ \ket{-} & M_X=1 \end{cases}$ ,
    and then broadcasts $\{M_X\}_{X\in\mathcal{P}}$ via classical channel

    \item Peripheral nodes $\mathcal{P} = \{A,B,C,\ldots\}$ and the relay node $\mathcal{R}$ apply recovery operation to the corresponding qubits: 
    $\mathcal{U}_{\text{rec}} = \bigotimes_{X \in \mathcal{P}_0} H_{X_1} \otimes \bigotimes_{X \in \mathcal{P}_1} (Z_{X_1}H_{X_1} )$ \\
    where partitions depend on measurement outcomes: 
    $\mathcal{P}_0 = \{X \in \mathcal{P} | M_X = 0\}$, $\mathcal{P}_1 = \{X \in \mathcal{P} | M_X = 1\}$
\end{enumerate}
\vspace{0.5em}
\hrule height 1pt
\end{figure*}

\begin{figure*}[t] 
\centering
\includegraphics[width=0.85\textwidth]{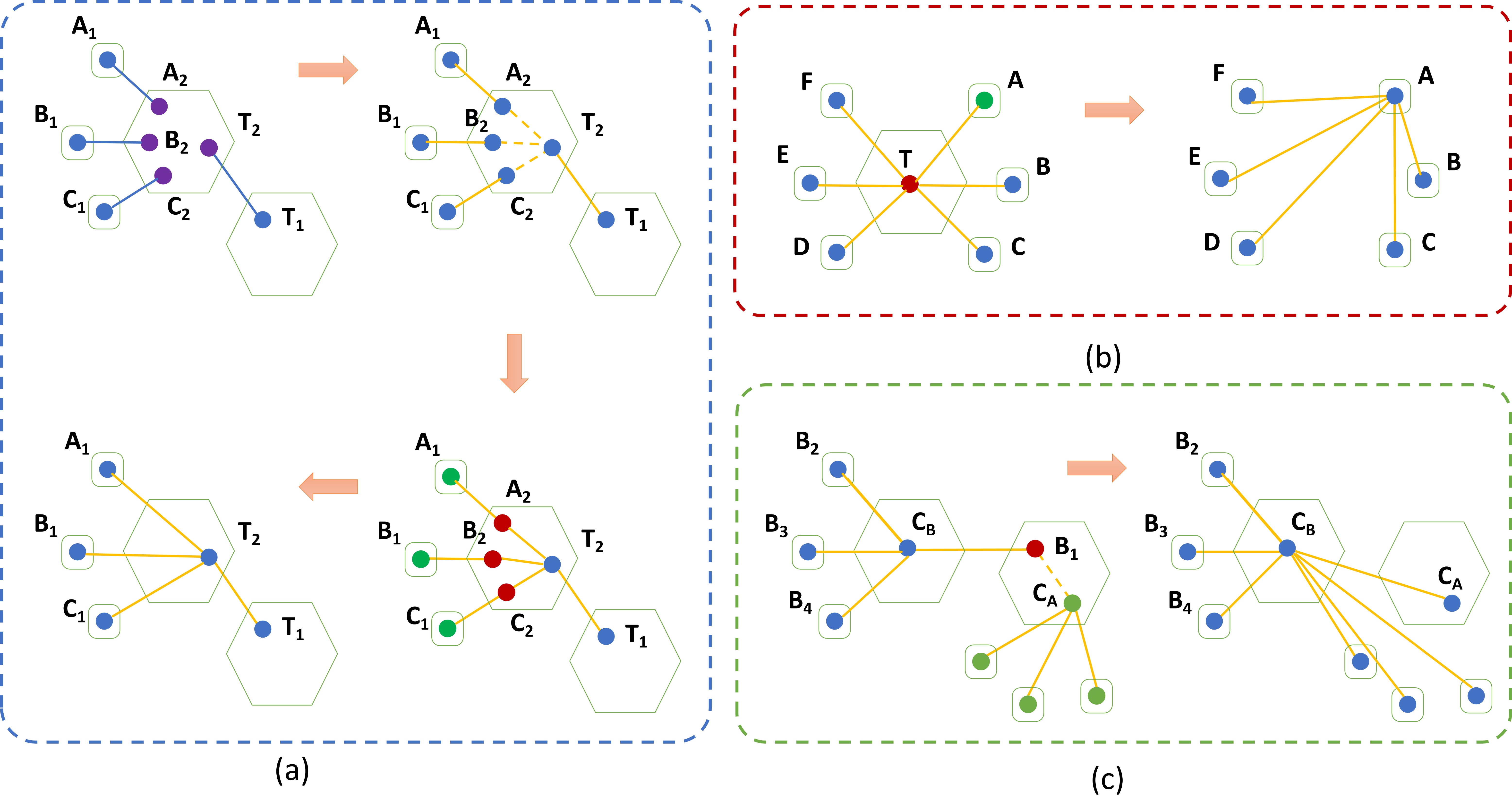}
\caption{
\textbf{Core topological operations governing graph state evolution.} 
(a) \textbf{Protocol 1: Subgraph Generation.} Relay nodes (hexagons) transform pre-shared Bell pairs (blue lines) into star-shaped subgraphs centered on qubit $T_2$ via local Hadamard gates (purple), CZ operations (yellow lines), and X-measurements (red).
(b) \textbf{Protocol 2: Center Transformation.} Dynamic migration of the topological center from node $T$ to $A$ is achieved by measuring $T$ (red) and applying conditional corrections on $A$ (green), without consuming additional entanglement.
(c) \textbf{Protocol 3: Star Fusion.} Two independent stars $G_A$ and $G_B$ are merged by entangling center $C_A$ with bridge node $B_1$, followed by the measurement of $B_1$. This integrates $G_A$ into $G_B$'s structure while preserving the star topology.
}
\label{fig2}
\end{figure*}

\subsection{Protocol~\ref{protocol:center_transform}: Graph Center Transformation Protocol}
\label{subsec:protocol2}
The complete procedural details are provided in Protocol~\ref{protocol:center_transform}, which endows the network with the flexibility to dynamically migrate topology centers on demand. By executing a Pauli $X$-basis measurement on the original central node and broadcasting the classical outcome, the designated successor node $A$ can re-establish the entanglement connectivity. This role transition is finalized by applying local unitary corrections—specifically $H_A$ or $Z_A H_A$—conditioned on the broadcast measurement results. The workflow is illustrated in Figure~\ref{fig2}(b), while the formal derivation is detailed in Appendix~\ref{app:protocol2}.

\begin{figure}[htpb]  
\small                
\hrule height 1pt
\vspace{0.5em}
\noindent \textbf{Protocol 2: Graph Center Transformation Protocol} 
\refstepcounter{protocol}\label{protocol:center_transform}
\vspace{0.5em}
\hrule height 0.5pt
\vspace{0.5em}

\textbf{Input}: 
\begin{itemize}
    \item Star graph state $|G\rangle$ generated by Protocol 1 with center $T$
    \item Designated new center node $A \in N(T)$ (neighborhood of $T$)
    \item Quantum resource: $|G\rangle = \left( \prod_{v \in N(T)} \cz_{T v} \right) |+\rangle^{\otimes |V|}$
    where:
    \begin{itemize}
        \item $V = \{T\} \cup N(T)$ is the vertex set
        \item $N(T) = \{A,B,C,\ldots\}$ denotes peripheral nodes
        \item $A$ is the designated new center
    \end{itemize}
\end{itemize}

\textbf{Output}: Star-shaped graph state with center $A$: $|G'\rangle = \left( \prod_{u \in N'(A)} \cz_{A u} \right) |+\rangle^{\otimes |V|-1}$ \\
where $N'(A) = N(T) \setminus \{A\}$ is the new neighborhood.

\begin{enumerate}
    \item The original center node $T$ undergoes Pauli-X basis measurement. The measurement outcome $M_T \in \{0,1\}$ determines the collapsed state: $\ket{\psi_T} = \begin{cases} \ket{+} & M_T=0 \\ \ket{-} & M_T=1 \end{cases}$ where $M_T$ is subsequently broadcast to all peripheral nodes $\{v \in N(T)\}$ for coordinated correction.

    \item Node $A$ is designated as the new central reference, with the recovery operator applied as: $\mathcal{U}_{\text{rec}} = \begin{cases} H_A  & M_T = 0 \\ Z_A H_A & M_T = 1 \end{cases}$
\end{enumerate}
\vspace{0.5em}
\hrule height 1pt
\end{figure}

\subsection{Protocol~\ref{protocol:star_fusion}: Star Graph State Fusion Protocol}
\label{subsec:protocol3}
The implementation details are provided in Protocol~\ref{protocol:star_fusion}. The fusion process begins by establishing an entanglement link between two independent star-shaped subgraphs $G_A$ (with center $C_A$ and peripheral set $N_A$) and $G_B$ (with center $C_B$ and peripheral set $N_B$) through a controlled-Z (CZ) operation applied between $C_A$ and a designated peripheral node $B_1 \in N_B$. A subsequent Pauli-X basis measurement on $B_1$, coupled with classical communication of the outcome $M_{B_1}$, decouples $B_1$ from the graph via projective collapse. Nodes $C_A$ and its neighborhood $N_A$ then perform conditional corrections based on $M_{B_1}$, which integrates the entire subgraph $G_A$ into $G_B$ as peripheral nodes, thereby subsuming it under the new center $C_B$. 

The resulting fused graph state $\left|G_{\text{fused}}\right\rangle$ is centered on $C_B$. Its expanded neighborhood set is defined as $N_{\text{fused}}(C_B) = N_A \cup \{C_A\} \cup (N_B \setminus \{B_1\})$. This fusion process preserves the star topology invariant throughout. Figure~\ref{fig2}(c) provides a schematic illustration of this workflow. And the complete mathematical proof of correctness and stability analysis are provided in Appendix~\ref{app:protocol3}.

\begin{figure}[htpb]  
\small                
\hrule height 1pt
\vspace{0.5em}
\noindent \textbf{Protocol 3: Star Graph State Fusion Protocol} 
\refstepcounter{protocol}\label{protocol:star_fusion}
\vspace{0.5em}
\hrule height 0.5pt
\vspace{0.5em}

\textbf{Input}: 
\begin{itemize}
    \item Two independent star-shaped graph states:
    \begin{align*}
    |G_A\rangle &= \left( \prod_{u \in N_A} \!\!\!\cz_{C_A u} \right) |+\rangle^{\otimes |V_A|} &
    |G_B\rangle &= \left( \prod_{v \in N_B} \!\!\!\cz_{C_B v} \right) |+\rangle^{\otimes |V_B|}
    \end{align*}
    where:
    \begin{itemize}
        \item $C_A$ and $C_B$ are central nodes
        \item $N_A = N(C_A) = \{A_1, A_2, \dots, A_m\}$ peripheral nodes of $C_A$
        \item $N_B = N(C_B) = \{B_1, B_2, \dots, B_n\}$ peripheral nodes of $C_B$
        \item Designated connection node $B_1 \in N_B$
    \end{itemize}
\end{itemize}
\textbf{Output}: Fused star-shaped graph state centered at $C_B$:
$
|G_{\text{fused}}\rangle = \left( \prod_{w \in N_{\text{fused}}(C_B)} \!\!\!\cz_{C_B w} \right) |+\rangle^{\otimes |V_{\text{fused}}|}
$
where:
\begin{itemize}
    \item Fused neighborhood: $N_{\text{fused}}(C_B) = N_A \cup \{C_A\} \cup (N_B \setminus \{B_1\})$
    \item Total node count: $|V_{\text{fused}}| = |V_A| + |V_B| - 1$
\end{itemize}

\begin{enumerate}
    \item Apply controlled-Z operation between node $C_A$ and node $B_1$: $\cz_{C_A B_1} |G_A\rangle \otimes |G_B\rangle$ to establish entanglement between the two graph states.

    \item Perform Pauli-X basis measurement on node $B_1$. The measurement outcome $M_{B_1} \in \{0,1\}$ determines the collapsed state: $\ket{\psi_{B_1}} = \begin{cases} \ket{+} & M_{B_1}=0 \\ \ket{-} & M_{B_1}=1 \end{cases}$

    \item Implement recovery operation conditioned on measurement outcome $M_{B_1}$: $\mathcal{U}_{\text{rec}} = \begin{cases} H_{C_A} & M_{B_1} = 0 \\ \left( \prod_{v \in N_A } \!\!\! Z_v \right) (Z_{C_A} H_{C_A}) & M_{B_1} = 1 \end{cases}$, 
    Subsequent to correction, node $B_1$ is decoupled from the system.\\
    
\end{enumerate}
\vspace{0.5em}
\hrule height 1pt
\end{figure}

\begin{figure*}[t] 
\centering
\includegraphics[width=0.85 \textwidth]{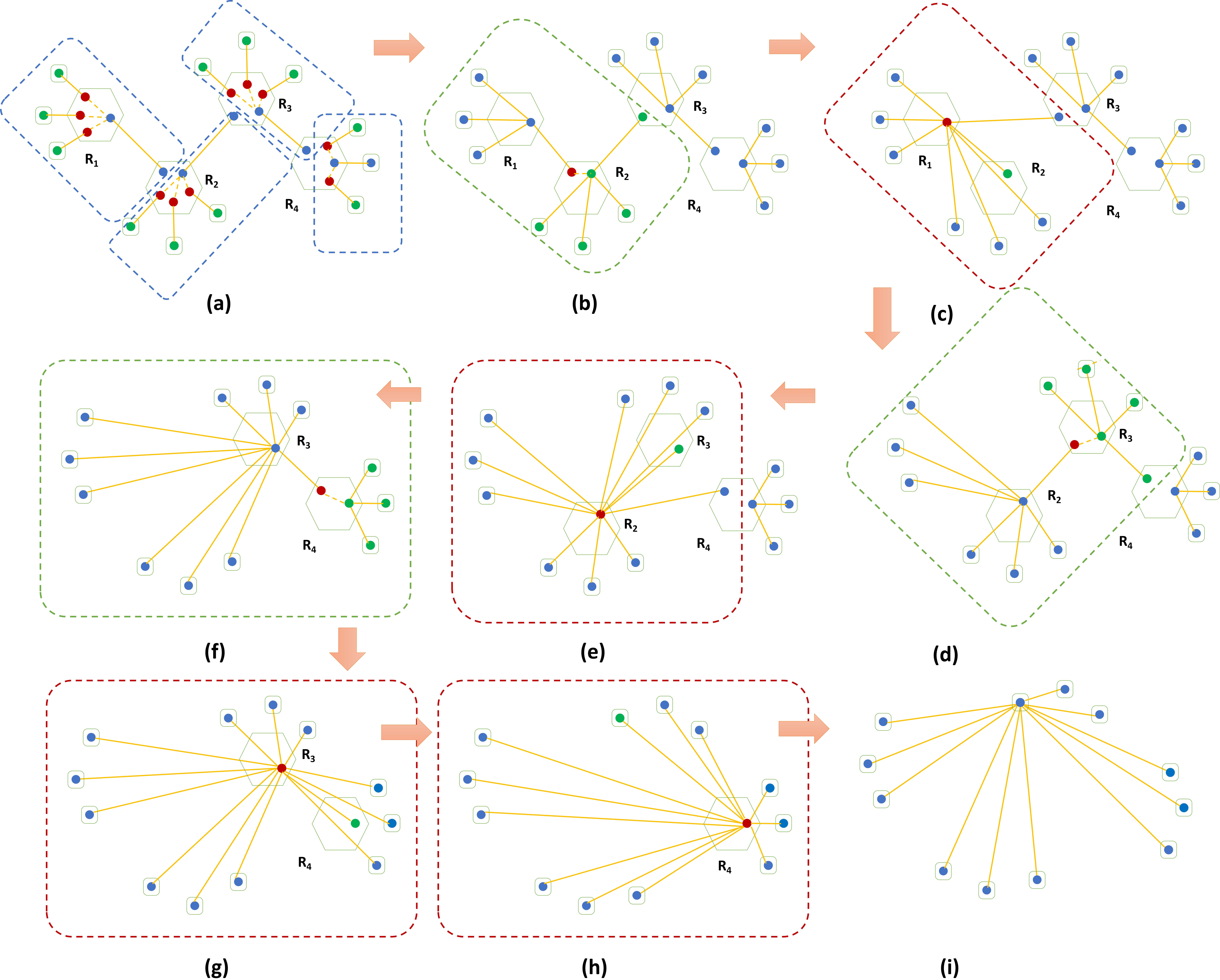}
\caption{
\textbf{End-to-end workflow for scalable star graph construction.}
(a) \textbf{Parallel Initialization:} Simultaneous execution of Protocol~\ref{protocol:star_gen} by all relay nodes to form local subgraphs (blue boxes).
(b) \textbf{Initial Fusion:} Protocol~\ref{protocol:star_fusion} merges the first two subgraphs via bridge node measurement (green box).
(c) \textbf{Center Migration:} Protocol~\ref{protocol:center_transform} shifts the center to a new relay node (red box).
(d)-(g) \textbf{Iterative Expansion:} Alternating execution of fusion and migration protocols progressively expands the topology.
(h)-(i) \textbf{Finalization:} The final center is designated among terminal nodes, and all remaining relay nodes are decoupled, yielding the target user-plane star graph (i).
}
\label{fig3}
\end{figure*}

\subsection{Integrated Workflow}
\label{subsec:integrated_workflow}
Through the coordinated execution of the three sub-protocols ~\ref{protocol:star_gen}--\ref{protocol:star_fusion}, we can achieve the gradual construction of large-scale star topologies. Taking the target topology in Figure ~\ref{fig1}(b) as a reference, Figure ~\ref{fig3} details its step-by-step evolution of the assembly process.

The workflow begins with the parallel generation of local subgraphs (Figure~\ref{fig3}(a)), during which all relay nodes synchronously execute the Protocol~\ref{protocol:star_gen}. The terminal nodes are driven to perform an $O(1)$ level local recovery operation through local Hadamard transforms, CZ gate operations, and classical feedforward signals generated by $Pauli-X$ measurement, thereby establishing the initial set of disjoint star subgraphs.

Subsequently, the process enters a dynamic iterative cycle of fusion and migration. As shown in Figure~\ref{fig3}(b), Protocol~\ref{protocol:star_fusion} merges adjacent subgraphs by establishing entanglement between a relay center and the peripheral nodes of neighboring subgraphs, followed by measurement-induced reduction to consolidate the structure. Immediately thereafter, Protocol~\ref{protocol:center_transform} (Figure~\ref{fig3}(c)) is invoked to migrate the topological center, preparing the system for the ensuing fusion step. 
This alternating sequence (Figures~\ref{fig3}(d)--(g)) gradually integrates all subgraphs into a large star-shaped graph.

In the final stage (Figure~\ref{fig3}(h)), the topology center is transferred to the designated terminal node by finally executing the Protocol~\ref{protocol:center_transform}. This step makes the target graph state (Figure ~\ref{fig3}(i)) entirely composed of terminal nodes.

\section{Simulation Results and Performance Analysis}
\label{sec:simulation}
We conducted simulations using Qiskit Aer to evaluate the protocol's performance. The simulations included five dimensions: resource utilization, temporal feasibility, noise analysis, scalability limit analysis, and scalability via erasure conversion.
The parameter settings used in the experiments were the latest experimental parameters of the dual-species trapped ion platform mapped by the protocol. Unless otherwise specified, the network size was set to $m \in [3, 12]$ relay nodes and $n \in [3, 6]$ peripheral nodes. For simulations involving noise, the results were the average of $N=2000$ Monte Carlo experiments.

\subsection{Resource Utilization}
\label{subsec:resource_efficiency}

We use two metrics to measure the resource utilization of the protocol. $\mathcal{C}_{\text{ent}}$ represents the number of Bell states consumed by a single qubit in the final state, and $\rho$ represents the number of quantum gates consumed by a single qubit in the final state.
Since this protocol has deterministic topological characteristics, we can obtain analytical expressions for these two metrics through analytical analysis. For $\mathcal{C}_{\text{ent}}$, assuming there are $m$ relay nodes in the network, each relay node has $n$ terminal nodes, and the required number of Bell states consists of relay-terminal links $mn$ and relay-relay links $m-1$. Therefore, the analytical expression for $\mathcal{C}_{\text{ent}}$ is as follows:

\begin{equation}
\mathcal{C}_{\text{ent}}(m,n) = \frac{N_{\text{links}}}{N_{\text{qubits}}} = \frac{mn + (m-1)}{mn} = 1 + \frac{1}{n} - \frac{1}{mn}.
\label{eq:analytical_cost}
\end{equation}

This derivation shows that resource overhead is primarily determined by the number of peripheral nodes $n$. As the network size increases ($m \to \infty$), the entanglement cost monotonically converges to an asymptotic lower bound of $1 + 1/n$.

\begin{figure*}[t] 
\centering
\includegraphics[width=0.85 \textwidth]{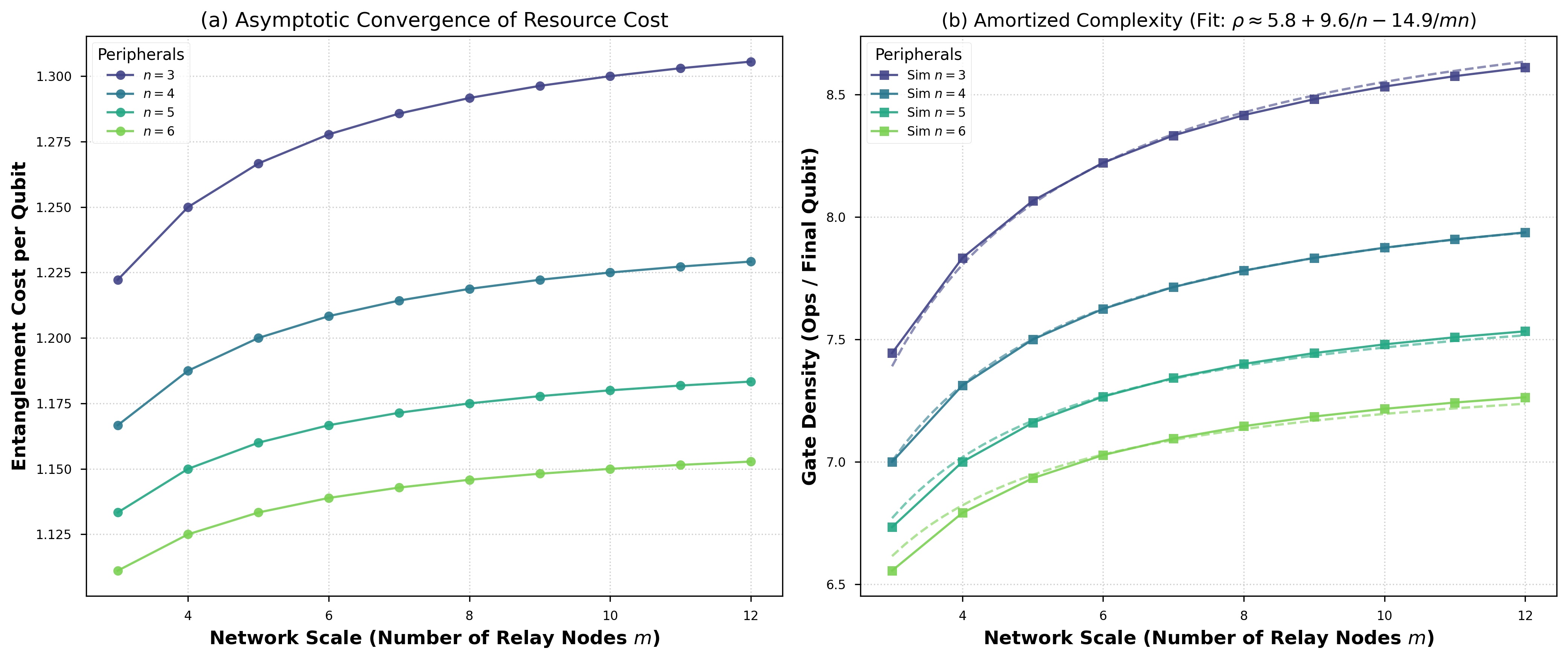}
\caption{\textbf{Resource Utilization.} (a) The entanglement cost per qubit, $\mathcal{C}_{\text{ent}}$, matches the prediction of the analytical formula (Eq.~\ref{eq:analytical_cost}) and converges to the asymptotic limits $\approx 1.33$ ($n=3$) and $\approx 1.16$ ($n=6$). (b) The gate density exhibits a bounded saturation mode, highly matching the prediction of the model $\rho \approx c_{\text{base}} + c_{\text{overhead}}/n - c_{\text{corr}}/mn$ (dashed line). This confirms the existence of initialization overhead at small scales, but the control complexity asymptotically saturates to a constant, thus verifying the amortized $O(1)$ scalability.} 
\label{fig:resource_complexity}
\end{figure*}

As shown in Figure ~\ref{fig:resource_complexity}(a), the numerical simulation results are in high agreement with the analytical expression predictions. Under different values of $n$, $\mathcal{C}_{\text{ent}}$ approximates different constant values. For example, when $n=6$, $\mathcal{C}_{\text{ent}}$ quickly stabilizes at $\approx 1.16$. Furthermore, the larger $n$ is, the smaller the corresponding constant value, indicating that the bell state resource overhead of the relay nodes introduced in this protocol is effectively distributed to each terminal node.

For $\rho$, we can see from Figure \ref{fig:resource_complexity}(b) that it is highly similar to Figure \ref{fig:resource_complexity}(a). We can establish the following model using a similar approach, assuming the total number of quantum gates required for the entire protocol process is $G_{\text{total}}$, which can be considered to consist of two parts. First, the number of quantum gates required in the relay-terminal interaction process, whose size is proportional to $mn$. Second, the number of quantum gates required in the relay-relay interaction process, whose size is proportional to $m-1$. Dividing the sum of the two terms by the total number of qubits in the final state $N=mn$, we can obtain the analytical model of $\rho$ as follows:
\begin{equation}
\rho(m,n) = \frac{G_{\text{total}}}{mn} \approx c_{\text{base}} + \frac{c_{\text{overhead}}}{n} - \frac{c_{\text{corr}}}{mn}.
\label{eq:gate_density_model}
\end{equation}

In this model, $c_{\text{base}}$ represents the intrinsic quantum gate depth lower bound of the protocol, $c_{\text{overhead}}$ represents the unit additional overhead introduced by the relay node, and $c_{\text{corr}}$ represents the boundary correction term. Experimental results show that the analytical expression and simulation data are in high agreement. As the network size increases $m \to \infty$, $\rho(m,n)$ will also gradually tend to a constant. In fact, the bounded convergence of the two metrics $\mathcal{C}_{\text{ent}}$ and $\rho(m,n)$ proves that our protocol has the scalability of amortized $O(1)$. This fundamentally alleviates the superlinear expansion bottleneck caused by global state synchronization and path search in other routing schemes.

\subsection{Temporal Feasibility}
\label{subsec:temporal_feasibility}

To evaluate the temporal feasibility of the protocol in wide-area quantum networks, specifically by analyzing the relationship between the protocol execution time and the decoherence time of the qubits, we constructed a time simulation model based on real-world physical experimental parameters. This model follows the workflow process shown in section \ref{subsec:integrated_workflow}, and the total execution time of the protocol, $T_{\text{exec}}$, can be expressed as:

\begin{equation}
T_{\text{exec}} \approx T_{\text{init}} + (N-1) \cdot (T_{\text{fusion}} + T_{\text{mig}}) + T_{\text{final}}.
\end{equation}

In this model, all physical parameters are set according to experimental parameters, as shown in Table \ref{tab:hardware_params}.

\begin{table}[htpb]
\caption{Physical hardware and network parameters utilized in the temporal feasibility model.}
\label{tab:hardware_params}
\footnotesize 
\begin{ruledtabular} 
\begin{tabular}{lll} 
\textbf{Symbol} & \textbf{Description} & \textbf{Value} \\ 
\hline 
$T_{1Q}$ & Single-qubit gate operation time & $10\,\mu\text{s}$ \\
$T_{2Q}$ & Two-qubit gate operation time & $100\,\mu\text{s}$ \\
$T_{\text{meas}}$ & Qubit measurement and readout time & $150\,\mu\text{s}$ \\
$T_{\text{proc}}$ & Classical signaling processing latency & $50\,\mu\text{s}$ \\
$L_{\text{link}}$ & Relay-to-Relay link distance & $50.0\,\text{km}$ \\
$L_{\text{access}}$ & Terminal-to-Relay link distance & $10.0\,\text{km}$ \\
$c_{\text{fiber}}$ & Effective signal propagation speed in fiber & $2 \times 10^5\,\text{km/s}$ \\
$T_{\text{coherence}}$ & Coherence limit of $^{171}\mathrm{Yb}^{+}$ & $1.0\,\text{s}$ \\ 
\end{tabular}
\end{ruledtabular} 
\end{table}

First, $T_{\text{init}}$ corresponds to the parallel generation stage of the local star subgraph. Its latency is composed of the local double-bit gate, measurement, terminal-relay link transmission latency $T_{\text{lat\_access}}$, and the single-bit gate operation that triggers correction, i.e., $T_{\text{init}} \approx T_{2Q} + T_{\text{meas}} + T_{\text{lat\_access}} + T_{1Q}$. Here, $T_{\text{lat\_access}} = (L_{\text{access}} / c_{\text{fiber}}) + T_{\text{proc}}$ includes the signal flight time of the terminal-relay link $L_{\text{access}} = 10.0$ km and the classical logic processing latency $T_{\text{proc}}$ within the node.

Second, the network mainly consists of $m-1$ sequential iterations between relay nodes. A single iteration consists of two consecutive processes: subgraph fusion ($T_{\text{fusion}}$) and center migration ($T_{\text{mig}}$). The subgraph fusion component $T_{\text{fusion}}$ encompasses the time for cross-node CZ gates, measurement, transmission of results via relay-relay links, and triggering local correction, i.e., $T_{2Q} + T_{\text{meas}} + T_{\text{lat\_backbone}} + T_{1Q}$; while the center migration component $T_{\text{mig}}$ characterizes the latency of measuring the current logical center and forwarding the result to the next-hop relay node to trigger subsequent operations, i.e., $T_{\text{meas}} + T_{\text{lat\_backbone}} + T_{1Q}$. Here, $T_{\text{lat\_backbone}}$ is composed of the fiber delay determined by the trunk-to-trunk link length $L_{\text{link}} = 50.0$ km and the classical processing time $T_{\text{proc}}$.

Finally, $T_{\text{final}}$ in the formula represents the delivery stage to the end user, whose delay is mainly affected by end-point measurements, terminal-to-trunk link feedback, and local logic correction, i.e., $T_{\text{final}} \approx T_{\text{meas}} + T_{\text{lat\_access}} + T_{1Q}$.

\begin{figure}[htbp]
\centering
\includegraphics[width=\columnwidth]{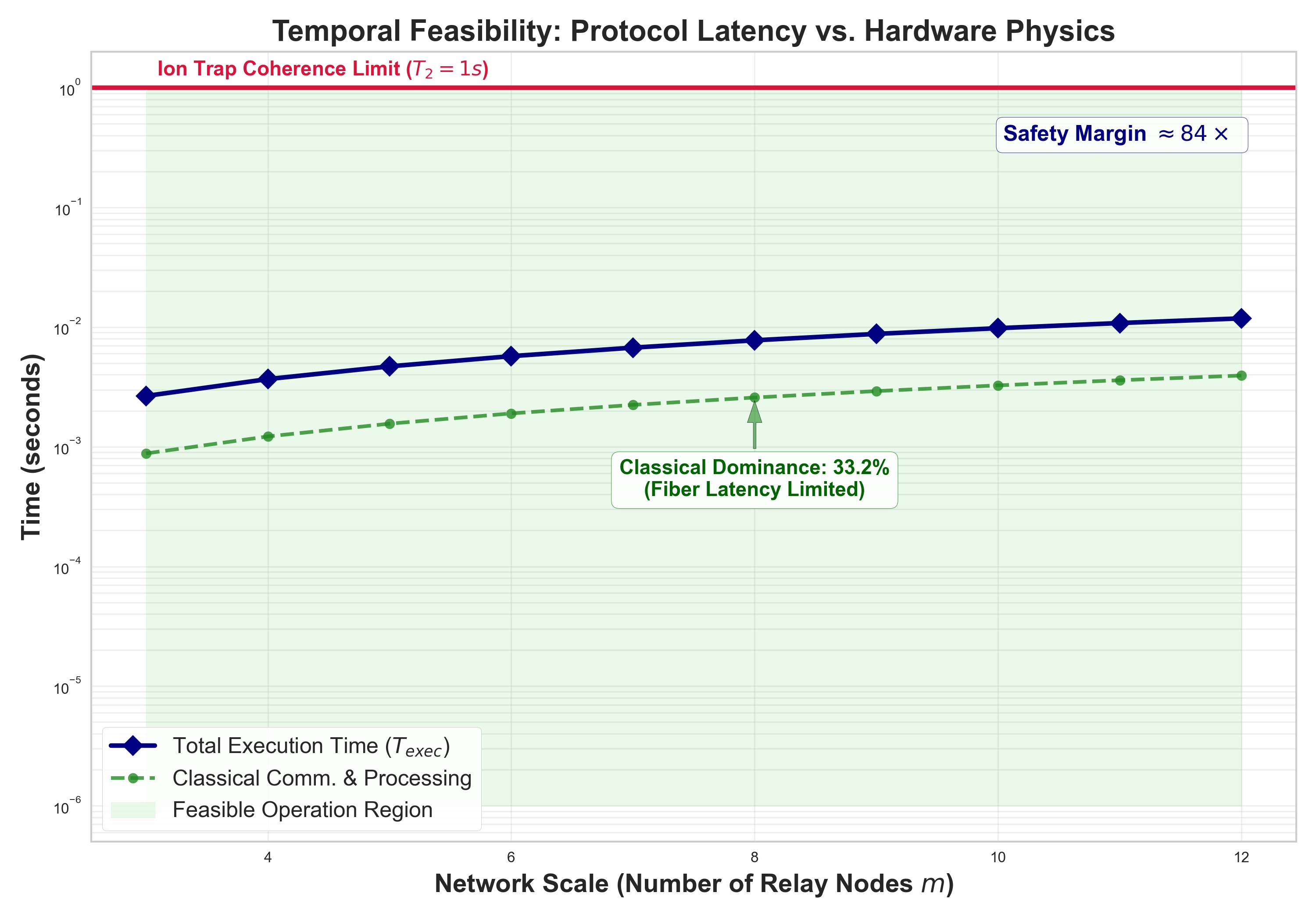}
\caption{\textbf{Temporal Feasibility.} The red solid line indicates the $1.0~\mathrm{s}$ coherence limit of $^{171}\mathrm{Yb}^{+}$ memory qubits. The dark blue solid line represents the total execution time of the protocol, and the green dashed line represents the latency caused by classical communication and processing. Simulation results show that the security margin is $\approx 80\times$, and classical communication time accounts for approximately 33\% of the total latency. This indicates that the protocol operates under low-latency conditions.}
\label{fig:temporal_feasibility}
\end{figure}

We use a conservative set of simulation parameters for evaluation: fiber propagation speed $c \approx 2 \times 10^5 ~ \mathrm{km/s}$ and processing latency $T_{\text{proc}}=50 ~ \mu \mathrm{s}$. Consistent with the two-species architecture defined in Section ~\ref{subsec:hardware}, we compare the total execution time to the coherence limit of $^{171}\mathrm{Yb}^{+}$ storage qubits $T_2^{\text{mem}} \approx 1.0 ~ \mathrm{s}$.
As shown in Figure ~\ref{fig:temporal_feasibility}, even for a large-scale network spanning 12 relay nodes, the total execution time of the protocol is still two orders of magnitude lower than the coherence limit, with a safety margin of $\approx 80 \times$.
Classical communication time accounts for approximately 33\% of the total latency, indicating that the protocol operates under low latency. This is due to the amortized $O(1)$ scalability of this protocol.
In traditional routing protocols using Pauli frames, nodes need to wait for global state synchronization information before performing recovery operations, thus increasing linearly with the network diameter and resulting in significant chain-like latency accumulation. Our protocol's amortized $O(1)$ complexity nature compresses the accumulated latency from potentially hundreds of milliseconds of global coordination waiting time to the local millisecond level, preventing the linear or non-linear accumulation of decoherence distortion with path length.

\subsection{Noise Analysis}
\label{subsec:noise_resilience}
In this subsection, we analyze the protol's performaocnce in noisy environments through simulation experiments. The simulation initially assumes high-fidelity predicted Bell pairs as initial resources, i.e., a purified state using $F_{\text{init}} \to 1$. This simplified setup allows us to focus on evaluating the impact of core error factors during protocol operation. These factors include quantum gate error $\mathcal{E}_{\text{gate}}$, quantum measurement operation error $\mathcal{E}_{\text{meas}}$, and the fidelity decrease caused by accumulated qubit decoherence. We construct the following hybrid model for the system's total fidelity $F_{\text{total}}$:

\begin{equation}
F_{\text{total}} \approx F_{\text{circuit}}(\mathcal{E}_{\text{gate}}, \mathcal{E}_{\text{meas}}) \times \exp\left(-\frac{T_{\text{net}}}{T_2^{\text{mem}}}\right)
\label{eq:hybrid_model}
\end{equation}

In this model, we employ a hybrid modeling approach combining stochastic circuit simulation and analytical network modeling. Here, $F_{\text{circuit}}$ represents the quantum gate error, the fidelity loss introduced by quantum measurement operation errors. The exponential term represents the qubit decoherence distortion introduced by the cumulative delay $T_{\text{net}}$ during protocol execution.
To ensure consistency in the analysis, this simulation uses the time delay parameters set by \ref{subsec:temporal_feasibility}. Notably we assume the probabilistic generation of entanglement is managed by the relay nodes as asynchronous buffers. Therefore, this model focuses on the evolution fidelity of entangled resources after a successful heralding event, without considering the attempt waiting time before link establishment.

\begin{figure*}[t] 
\centering
\includegraphics[width=0.95\textwidth]{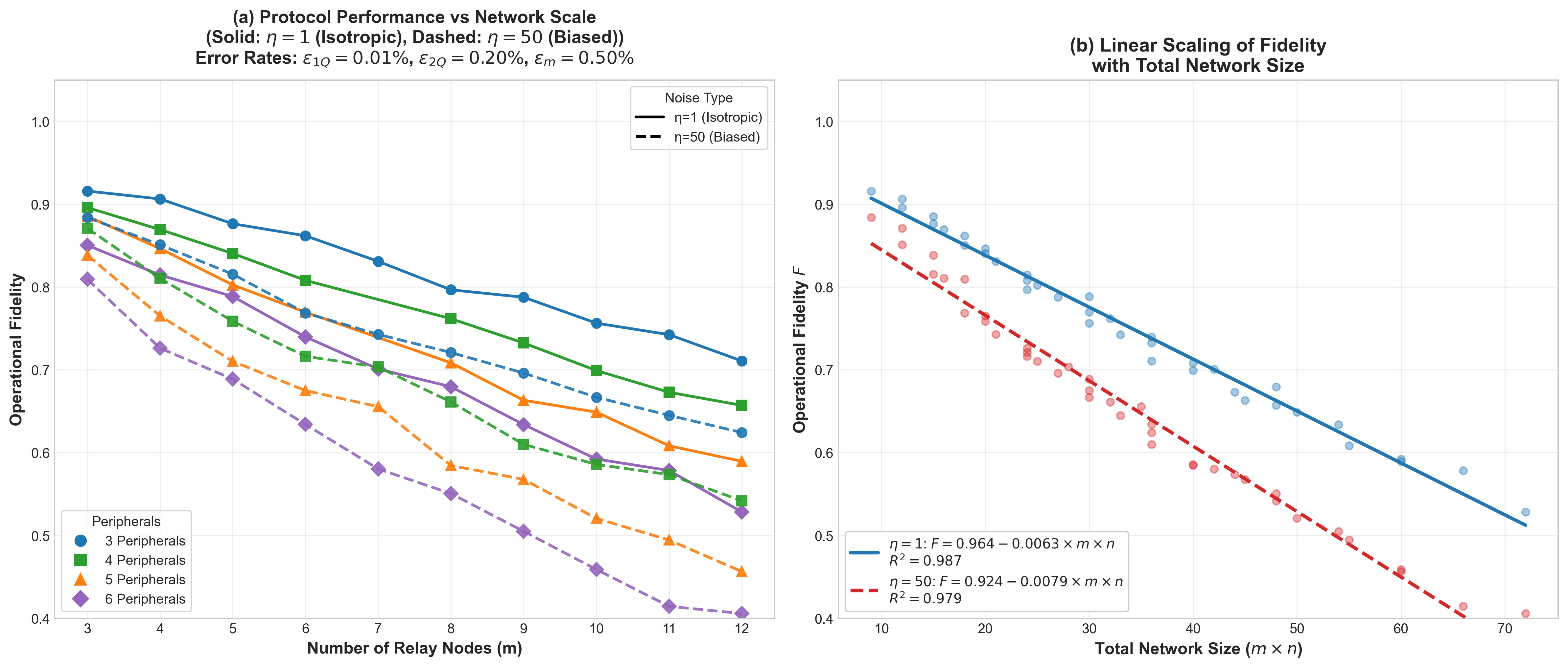}
\caption{\textbf{Fidelity scaling under isotropic vs. biased noise.} (a) Simulation results of entanglement fidelity evolution with network size ($m×n$). Under both the standard isotropic noise ($\eta=1$, solid line) and the experimentally motivated Z-bias noise ($\eta=50$, dashed line) models, the system fidelity shows a significant monotonic linear decline. (b) Linear regression ($R² \approx 0.98$) analysis confirms the additive nature of error accumulation. Although a strong Z-bias reduces the absolute fidelity intercept due to measurement sensitivity, the highly consistent linear expansion characteristics under both noise models demonstrate that this protocol has extremely strong performance predictability and structural scalability under different physical noise environments.}

\label{fig:noise_resilience}
\end{figure*}

\subsubsection{Anisotropic Error Model}

Unlike the traditional isotropic depolarization assumption, we model the noise characteristics unique to trapped ion architectures. The noise channel $\mathcal{E}$ is defined as an anisotropic Pauli channel:

\begin{equation}
\mathcal{E}(\rho) = (1 - p)\rho + \sum_{\sigma \in \{X,Y,Z\}} p_\sigma \sigma \rho \sigma,
\label{eq:biased_noise}
\end{equation}

The error probabilities satisfy $\sum p_\sigma = p$. To simulate the dominant role of magnetic field fluctuations and laser phase noise relative to spontaneous emission, we introduce a strong bias factor $\eta = p_Z / p_{X,Y} = 50$. Simulation parameters are based on current hardware benchmarks: single-qubit gate error $\epsilon_{1Q} = 0.01\%$, two-qubit gate (CZ) error $\epsilon_{2Q} = 0.2\%$, and measurement readout error $\epsilon_m = 0.5\%$.

\subsubsection{Error robustness and scalability analysis}

Simulation results illustrated in Figure~\ref{fig:noise_resilience} reveal the impact of noise structure on protocol performance. 
Notably, under the same total error rate, the absolute fidelity of the system is lower in a bias noise environment with $\eta=50$ than in an isotropic noise model with $\eta=1$. This bias noise sensitivity stems from the protocol's dependence on Pauli-X measurements. Essentially, since the Pauli $Z$ operator anti-commutes with the measurement basis, which is represented by the relation $Z|+\rangle = |-\rangle$, $Z$ errors directly manifest as logical bit-flips in the readout results.

However, regardless of the noise bias condition, the fidelity exhibits a highly consistent linear decreasing trend. Based on this observation, we propose a linear phenomenological model to describe the fidelity decay law as follows:

\begin{equation}
F_{m,n} \approx F_0 - \gamma \cdot mn.
\label{eq:linear_model}
\end{equation}

In this expression, parameter $F_0$ represents the baseline fidelity of the system at the current noise level, while parameter $\gamma$ represents the marginal fidelity loss rate introduced as the network size expands. The network size is defined as the product of the number of relays $m$ and the number of branches $n$. This model contains a basic assumption that error accumulation is a simple additive process and that no complex nonlocal error propagation occurs.

To verify the effectiveness of the model, we performed linear regression analysis on the simulation data. The fitting results show that the linear model exhibits extremely high explanatory power in both scenarios, with a goodness of regression of $R^2 \approx 0.98$. The specific fitting parameters are compared below. Under ideal isotropic noise ($\eta=1$), the model evolves to $F(m,n) \approx 0.964 -0.0063 \times (m \times n)$, while under actual biased noise ($\eta=50$), the model evolves to $F(m,n) \approx 0.924 - 0.0079 \times (m \times n)$.

A comparison of these two datasets reveals that a strong $Z$-bias not only reduces the initial baseline fidelity $F_0$ but also slightly accelerates the decay rate $\gamma$ as the network scale increases. This quantified linear correspondence confirms the predictability of the protocol's performance in large-scale network topologies.

\subsection{Scalability Limit Analysis}
\label{subsec:phase_diagram}

The analysis of temporal analysis in the subsection \ref{subsec:resource_efficiency} shows that in the physical platform mapped by our protocol, the protocol execution latency is much lower than the coherence time of the qubit, and the decoherence distortion has a limited impact on the overall fidelity. In contrast, at the current experimental level, the measurement error rate $\epsilon_m$ remains high, which is the main factor limiting the scaling of quantum networks. To explore this core bottleneck in depth, we use the measurement error rate $\epsilon_m$ as the independent variable and perform a parameter scan in the logarithmic scale $\epsilon_m \in [10^{-3}, 10^{-1}]$. In the simulation experiment, the remaining physical parameters were set as follows: the single-bit gate error rate was $\epsilon_{1Q} = 0.01\%$, the double-bit gate error rate was $\epsilon_{2Q} = 0.2\%$, the network size was set to the number of relay nodes $m \in [3, 12]$, and the number of terminal nodes connected to each relay node was fixed at $n=4$. Figure~\ref{fig:phase} illustrates two performance ranges with distinctly different physical meanings:

\begin{figure}[htbp]
\centering
\includegraphics[width=\columnwidth]{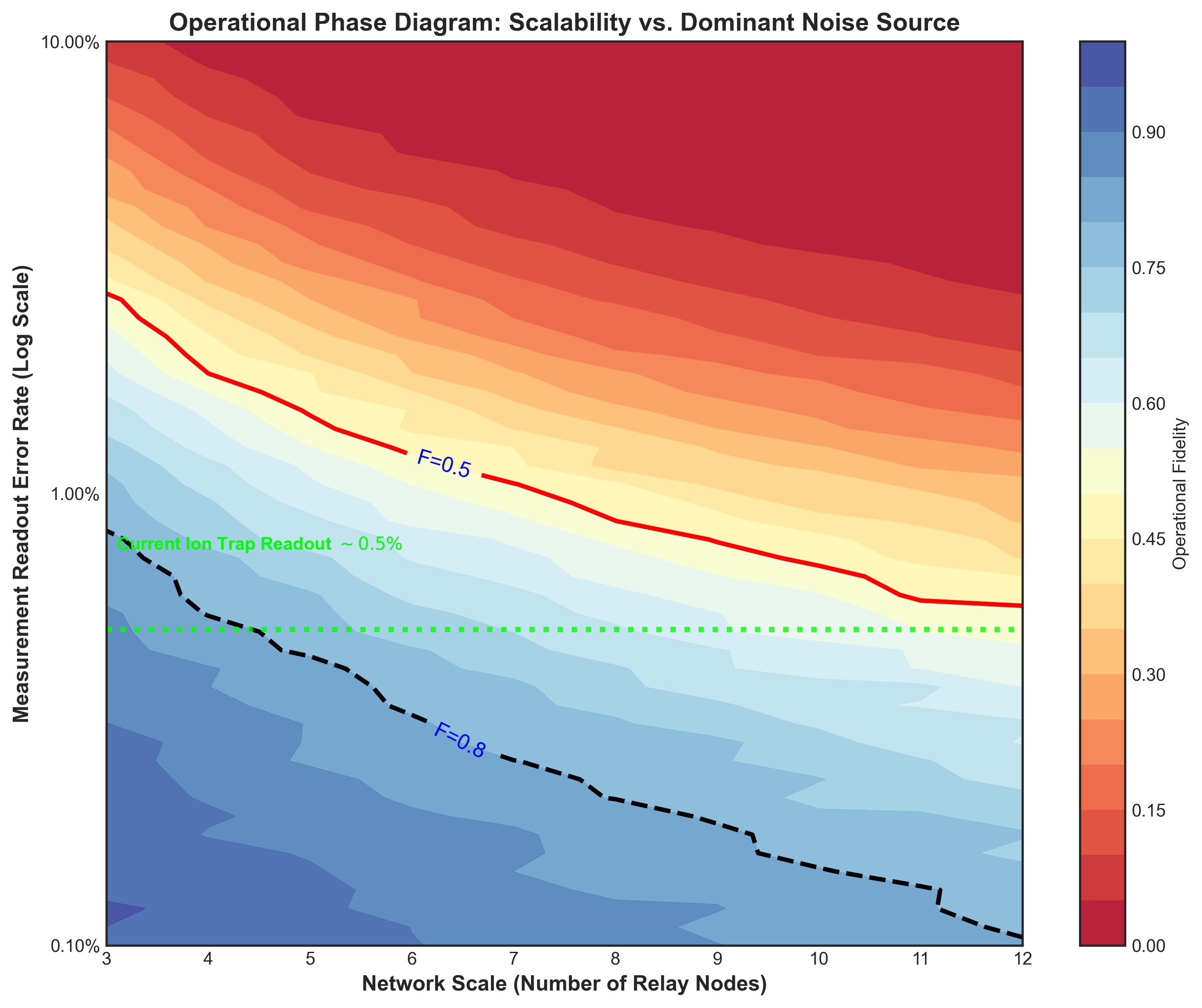}
\caption{\textbf{Operational Phase Diagram.} The relationship between global fidelity and network scale ($m$) and readout error ($\epsilon_m$) is illustrated in the figure. The contour line $F=0.5$ (red) defines the boundary where entanglement exists, while $F=0.8$ (black dashed line) determines the threshold for high-fidelity applications. The intersection of the current technical baseline ($\epsilon_m = 0.5\%$, green dashed line) and the survival boundary is located at $m \approx 12$. The linear characteristics of the isofidelity contour lines in the figure, rather than the exponential decay corresponding to the coherence limit, demonstrate that low-latency protocols have decoupled system performance from qubit decoherence.}
\label{fig:phase}
\end{figure}

\begin{itemize}
    \item \textbf{High-Fidelity Regime ($F>0.8$):} defines the high-fidelity entanglement required for logical qubit synthesis and fault-tolerant protocols. With current state-of-the-art readout technology ($\epsilon_m = 0.5\%$), this interval is limited to small clusters ($m \approx 4$). To scale the network to approximately $m=10$ nodes, the readout error needs to be kept below $0.15\%$.
    \item \textbf{GME Persistence Regime ($0.5 < F < 0.8$):} The solid red line in the figure marks the theoretical boundary for the existence of true many-body entanglement. When the readout error is $\epsilon_m = 0.5\%$, the protocol can still maintain effective entanglement in a wide area network with $m=12$ relay nodes.
    
\end{itemize}

Simulation results, through the linear slope of the isofidelity curve in a logarithmic coordinates, confirm the scaling relationship between the maximum network size $m_{\text{max}}$ and the negative logarithm of the measurement error rate $\epsilon_m$, which is $m_{\text{max}} \propto -\log \epsilon_m$. This finding quantitatively verifies the previous hypothesis that on the Trapped-ion ion trap platform mapped by our protocol, thanks to the ample time margin provided by the amortization of $O(1)$ complexity, the scalability bottleneck of the quantum network has successfully shifted from the constraint of qubits coherence time to quantum measurement accuracy. This means that we can focus on improving measurement fidelity, such as trying to combine fault-tolerant techniques like erasure conversion in our protocol, which is the subject of the next subsection \ref{sec:erasure}.

\subsection{Scalability via Erasure Conversion}
\label{sec:erasure}
Analysis of subsection \ref{subsec:phase_diagram} shows that measurement readout fidelity is the core bottleneck restricting the final graph state fidelity. In our protocol, the accuracy of Pauli-$X$ measurements directly determines the final quality of the globally entangled state. To address this challenge, we utilize dual-threshold fluorescence detection technology ~\cite{refera1, refera2} on a dual-species platform to transform hidden bit-flip errors into heralded erasure errors. 
Based on the aforementioned erasure conversion mechanism, and taking into account the symmetry and algebraic locality of the star graph state, we perform the following engineering optimizations on specific aspects of the Pauli-$X$ measurement in the three sub-protocols.

\begin{figure}[htbp]
\centering
\includegraphics[width=0.9\linewidth]{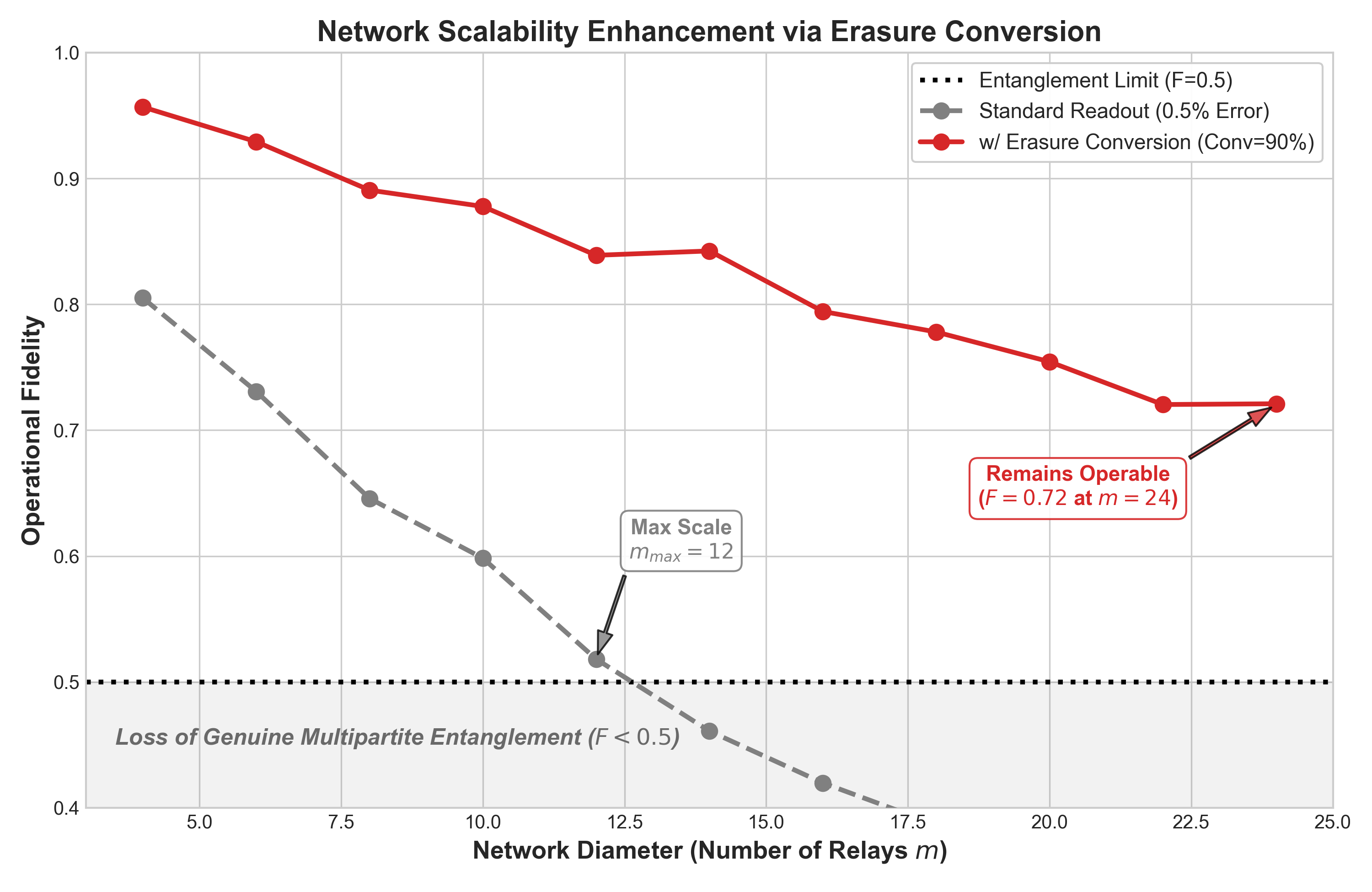} 
\caption{\textbf{Network Scalability Enhancement through Erasure Conversion}
The figure compares the evolution trend of system fidelity with scale under the settings of $n=4$ and $\epsilon_m = 0.5\%$. The gray dashed line and the red solid line correspond to the standard readout scheme and the erase-conversion scheme with a conversion efficiency of $\eta_{\text{conv}}=90\%$, respectively. The results show that the erase-conversion technique can reduce the effective error rate by an order of magnitude, thus enabling the system to maintain a high fidelity $F > 0.7$ even after crossing 24 relay nodes.}
\label{fig:erasure_extension}
\end{figure}

During the execution of Protocols~\ref{protocol:star_gen} and~\ref{protocol:star_fusion}, we introduce a local retry strategy. When a relay node detects a Pauli-$X$ measurement triggering an erase error warning, it immediately invokes pre-deployed spatial reuse resources to locally retry the $CZ$ gate and measurement operations on the new Bell pair. This mechanism ensures that each logical link obtains highly deterministic Pauli-$X$ measurement results, thereby improving the fidelity of the final graph.
For Protocol~\ref{protocol:center_transform}, we optimized the original measurement-based scheme into a logic gate-based scheme. By applying local Hadamard gates to the original center node and its target neighboring nodes, the transfer of the center node can still be achieved. This higher-precision quantum gate operation replaces the relatively low-precision measurement operation, avoiding error propagation caused by measurement errors. The auxiliary particles located on the relay nodes that are not measured do not need to be measured immediately; instead, they are retained as logical leaf nodes until the end of the process, and finally stripped through a Pauli-$Z$ measurement. Since in the star diagram, the $Z$ measurement error of the leaf node only produces a local phase effect and does not cause the collapse of the global entanglement structure, this strategy greatly reduces the substantial interference of measurement noise on the quality of the entangled state between end users.

We evaluate the performance improvement of the above optimization strategy by introducing an erasure transition efficiency $\eta_{\text{conv}}$. In this model, the effective logical error rate is suppressed to $\epsilon_{\text{eff}} = \epsilon_{m} \times (1 - \eta_{\text{conv}})$. Under this modeling, the reduction in the effective logical error rate $\epsilon_{\text{eff}}$ comes at the cost of consuming more bell states. It should be noted that the entanglement cost derived in Sec.~\ref{subsec:resource_efficiency} represents the theoretical lower bound of the resources required to construct the target topology. After introducing the erasure transition, the actual physical Bell pair consumption will be affected by the success rate. However, since the current measured error rate is at a low level $\epsilon_m \approx 0.5\%$, the resulting additional retry overhead is only a marginal burden and will not change the asymptotic trend of resource costs as the network size scales.

As shown in Figure ~\ref{fig:erasure_extension}, with a configuration of $n=4$ terminals per relay, the erasure transition reduces the logical error rate by an order of magnitude. Even when the network diameter expands to $m = 24$ relay nodes, the global fidelity remains above $F = 0.7$.
This design is essentially a trade-off between resources and quality: by consuming some redundant resources to offset erasure losses, it filters out hidden logic errors. This makes the network deliver not deterministic but noisy outputs, but high-quality heralded resources. In fact, this high-quality probabilistic resource is actually more attractive because fault-tolerant and error-correcting frameworks are much more efficient at handling "explicit losses" than "unknown flips".

\section{Comparison With Related Work}
\label{sec:comparison}

To quantitatively evaluate the overall performance of this protocol, we conducted a rigorous comparison with three cutting-edge quantum routing protocols: the hypergraph-optimized rate distribution by Fan et al.~\cite{refFan}, the multipath fidelity routing by Sutcliffe et al.~\cite{refevan}, and the gate-reduction schemes (CKL) by Chelluri et al.~\cite{refCKL}. A quantitative summary of these protocols is provided in Table~\ref{tab:comparison}.

\begin{table*}[t]   
\caption{Architectural comparison against state-of-the-art benchmarks.}
\label{tab:comparison}
\centering
\begin{tabular}{lcccc}
\toprule
\textbf{Feature} & \textbf{Our Work} & \textbf{Fan et al.}~\cite{refFan} & \textbf{CKL et al.}~\cite{refCKL} & \textbf{Sutcliffe}~\cite{refevan} \\
\midrule
Optimization Goal & Latency \& Isolation & Generation Rate & Gate Count & Multipath Reliability \\
Decision Complexity & $O(1)$ (Local Logic) & NP-Hard (LP) & $O(V+E)$ (BFS) & $O(V|E|^2)$ (Max-Flow) \\
Control Mode & Local Collapse & Global Sync & Global Pauli Tracking & Link-State Sync \\
Hardware Model & Heterogeneous & Homogeneous & Homogeneous & Homogeneous \\
Entanglement Cost & $1+1/n \approx 1$ & LP-Dependent & $(N-1)/N \approx 1$ & $\gg 1$ (Multipath) \\
\bottomrule
\end{tabular}
\end{table*}

\subsection{From Global Optimization to Local Collapse}

In dynamic quantum networks, the real-time speed of routing decisions depends mainly on computational efficiency. The multi-path scheme proposed by Sutcliffe et al.~\cite{refevan} uses min-cost max-flow algorithms; however, its polynomial complexity, $O(V|E|^2)$, has difficulty maintaining fast response times in large-scale topologies. Similarly, while the hypergraph flow optimization framework by Fan et al.~\cite{refFan} looks for optimal generation rates, the number of hyperedges in its models grows as $O(n^3 m^4)$. In large networks, such computational delays lead directly to decoherence because quantum states must wait for routing decisions. The CKL scheme~\cite{refCKL} addresses this by using Breadth-First Search (BFS) to make path-finding linear ($O(V+E)$), successfully reaching the theoretical resource lower bound of $(N-1)/N$. However, the fusion steps in CKL still require global Pauli frame tracking. Correction operators must be collected along long chains and broadcast to the whole network, causing high synchronization delays and heavy signaling loads in wide-area networks.

Our protocol adopts a completely different technical approach. Based on the algebraic properties revealed by the theorem ~\ref{thm:x-simplify}, we utilize the neighborhood independent set condition to allow the secondary operator caused by the measurement to collapse directly locally. This makes routing decisions no longer need to be aware of the entire network topology, but degenerates into a local operation that only requires neighborhood information. This protocol reduces the decision complexity to an amortized $O(1)$ level. Therefore, no matter how the network scales, the decision process can be completed in a very short time. This extremely low latency ensures that all classical signals can be delivered before the qubits decoherent ($T_{\text{calc}} \ll T_2$), thus avoiding the superlinear expansion problem common in traditional routing protocols from the underlying architecture.
Note that while the sequential topological expansion of the entire network globally requires $O(N)$ operations, our protocol guarantees that for any individual qubit, the local decision latency and the classical feedforward depth are strictly bounded by $O(1)$. This decouples the quantum memory waiting time from the network diameter.

\subsection{Hardware-Aware Design and Engineering Robustness}
Current quantum routing protocols~\cite{refFan, refevan, refCKL} generally assume that network nodes are functionally homogeneous. That means each node is responsible for both storing entangled states and handling frequent link measurements. However, in practical hardware architectures, such as trapped-ion systems, this design leads to severe physical conflicts, and noise generated by measurement operations can easily destroy stored bits within the same node through photon scattering. 
Our proposed protocol, through a heterogeneous logical role separation design, localizes measurement operations at relay nodes, avoiding the impact of measurement scattering noise on the stored qubits of terminal nodes.

Meanwhile, this design allows the protocol to perform high-frequency local retries without worrying about additional noise interference to terminal data during the retries. Therefore, combined with the localized characteristics of star topology, the protocol can incorporate erasure conversion techniques, effectively improving the fidelity of the protocol in noisy environments. This heterogeneous design inherently incorporates error suppression, demonstrating the protocol's excellent engineering robustness and providing a practical technical path for realizing high-fidelity quantum networks.

\section{Conclusion and Outlook}
\label{sec:conclusion}
This paper addresses the contradiction between probabilistic long-range entanglement generation and finite coherence time in large-scale quantum networks. Based on local measurement and classical feedforward, a routing protocol with amortized $O(1)$ decision complexity is proposed.

We map the protocol onto a dual-species Trapped-ion platform and perform simulation analysis. In terms of resource overhead, the bounded convergence of entanglement cost and gate density verifies the asymptotic scalability of the protocol. Regarding time feasibility, due to the amortized $O(1)$ scalability of this protocol, the safety margin of the total execution time reaches $\text{Margin} \approx 80\times$, and classical communication time accounts for approximately 33\% of the total latency, indicating that the protocol operates under low latency. Noise analysis further reveals that even under strong $Z$-bias noise, the fidelity expansion of the protocol always follows a linear scaling law. Furthermore, on the physical platform defined in this paper, since the $O(1)$ control logic provides sufficient time safety margin, the upper limit of quantum network expansion is no longer limited by the coherence time of the qubits, but mainly depends on the accuracy of a single physical measurement. To mitigate this measurment readout bottleneck, we leverage the branch independence of a star topology, transforming bit-flip errors into predictive erasure through an erasure transformation technique, and combining this with spatial multiplexing for in-situ retries. At a transformation efficiency of 90\%, this scheme reduces the effective logic error rate by an order of magnitude. This engineering strategy of "trading resources for fidelity" provides an important reference for the protocol design of future quantum internets.

Looking ahead, we will further explore extending this protocol to measurement-based or fusion-based quantum computing tasks, investigating its application potential in building scalable, fault-tolerant distributed quantum computing networks.

\section*{Data Availability}
The Qiskit simulation scripts and numerical data that support the findings of this study are openly available at: \url{https://github.com/ky87zheng/Quantum-Routing-Simulation}.

\bibliographystyle{unsrt} 
\bibliography{references}






\appendix

\section{Graph State Properties and Proofs}
\label{app:graph state}

\subsection*{A.1 Graph State Definition}
A graph state $|G\rangle$ is generated by preparing each vertex in the $|+\rangle$ state and applying controlled-Z ($\CZ$) gates along the edges of an undirected graph $G = (V,E)$, where $V$ represents qubits and $E$ denotes entangling connections:

\begin{equation}
|G\rangle = \prod_{(a,b)\in E} \CZ_{a,b} |+\rangle^{\otimes |V|},\quad 
|+\rangle = \frac{1}{\sqrt{2}}(|0\rangle + |1\rangle)
\label{eq:A1}
\end{equation}

\subsection*{A.2 Local Complementation}
Local complementation at vertex $i$, denoted $\text{LC}_i(G)$, modifies the graph topology through symmetric difference of edge sets:
\begin{equation}
G \mapsto (V, E \triangle \mathcal{E}_i), \quad 
\mathcal{E}_i = \{ (u,v) \mid u,v \in \mathcal{N}(i), u \neq v \}
\label{eq:lc_graph}
\end{equation}
This graph transformation corresponds to the unitary operation:
\begin{equation}
U_{\text{LC}_i} = R_x(\pi/2)_i \bigotimes_{j \in \mathcal{N}(i)} R_z(-\pi/2)_j
\label{eq:lc_unitary}
\end{equation}
where the rotation operators are defined as $R_{x,z}(\theta) = \exp(-i\theta \sigma_{x,z}/2)$.

Fig.~\ref{fig:8} demonstrates the graph transformation induced by local complementation (LC) operation. For instance, consider vertex $1$ with neighbor set $\mathcal{N}(1) = \{2,3,4\}$ and initial edge set $E = \{(1,2),(1,3),(1,4),(2,3),(2,4)\}$. The LC operation at vertex $1$ induces the transformation:
\begin{equation}
\LC_1(G): \quad E \mapsto E \triangle \{(3,4)\} = \{(1,2),(1,3),(1,4),(3,4)\}
\end{equation}

\begin{figure}[htbp]
    \centering
    \includegraphics[width=\columnwidth]{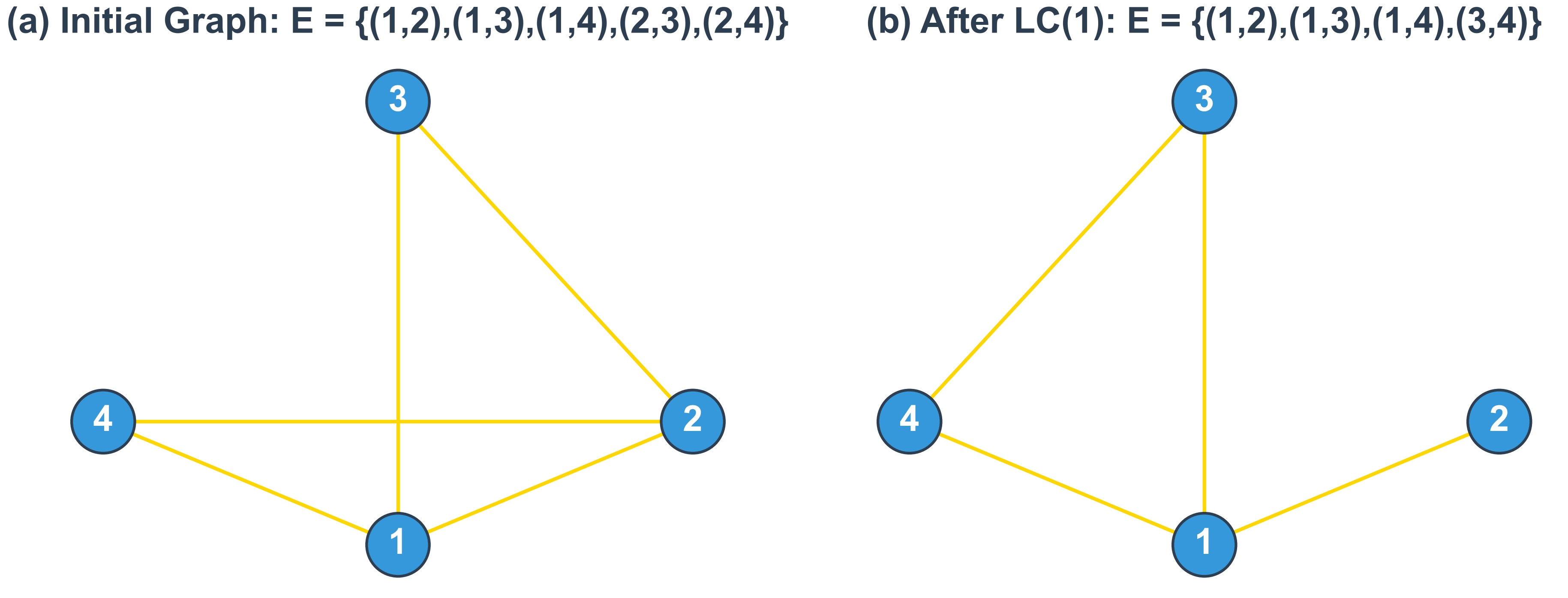}
    \caption{Graph transformation demonstration of Local Complementation operation.}
    \label{fig:8}
\end{figure}

This graph operation corresponds to the unitary evolution:
\begin{equation}
\LC_1|G\rangle = \left[ R_x(\tfrac{\pi}{2})_1 \bigotimes_{j=2}^4 R_z(-\tfrac{\pi}{2})_j \right] |G\rangle
\label{eq:lc_unitary1}
\end{equation}

\subsection*{A.3 Pauli Measurement Rules}

\subsubsection*{A.3.1 Z-basis Measurement}
Measurement of qubit $i$ in the Z-basis $\{|0\rangle, |1\rangle\}$ projects the graph state onto:
\begin{equation}
M_z^{(i)} |G\rangle = \frac{1}{\sqrt{2}} \left( |0\rangle_i \otimes U_0^{(i)} |G_{-i}\rangle + |1\rangle_i \otimes U_1^{(i)} |G_{-i}\rangle \right)
\label{eq:A4}
\end{equation}
where $|G_{-i}\rangle$ denotes the graph state after removing vertex $i$, and the measurement-induced corrections are given by $U_0^{(i)} = I^{\otimes|\mathcal{N}(i)|}$ and $U_1^{(i)} = Z^{\otimes|\mathcal{N}(i)|}$. The post-measurement state differs from the standard graph state $|G_{-i}\rangle$ by a recovery operation applied to the neighborhood $\mathcal{N}(i)$.

\subsubsection*{A.3.2 Y-basis Measurement}
Measurement of qubit $i$ in the Y-basis $\{|+y\rangle, |-y\rangle\}$ where $|\pm y\rangle = \frac{|0\rangle \pm i|1\rangle}{\sqrt{2}}$ results in:
\begin{equation}
M_y^{(i)} |G\rangle = \frac{1}{\sqrt{2}} \left( |+y\rangle_i \otimes U_{+y}^{(i)} |(\LC_i(G))_{-i}\rangle + |-y\rangle_i \otimes U_{-y}^{(i)} |(\LC_i(G))_{-i}\rangle \right)
\label{eq:A6}
\end{equation}
where $|(\LC_i(G))_{-i}\rangle$ represents the graph state after local complementation at vertex $i$ followed by vertex removal. The measurement-induced corrections are given by:
\begin{equation}
U_{\pm y}^{(i)} = R_z(\pm\pi/2)^{\otimes|\mathcal{N}(i)|}
\label{eq:A7}
\end{equation}
This corresponds to a coherent superposition of two graph modifications: the measured vertex collapses to a Y-basis state, while its neighborhood acquires conditional phase rotations correlated with the measurement outcome.

\subsubsection*{A.3.3 X-basis Measurement}
Measurement of qubit $i$ in the X-basis $\{|+\rangle, |-\rangle\}$ where $|\pm\rangle = \frac{|0\rangle \pm |1\rangle}{\sqrt{2}}$ gives:
\begin{equation}
M_x^{(i)} |G\rangle = \frac{1}{\sqrt{2}} \left( |+x\rangle_i \otimes U_+^{(i,j)} |G''_{-i}\rangle + |-x\rangle_i \otimes U_-^{(i,j)} |G''_{-i}\rangle \right),
\label{eq:A8}
\end{equation}
where $|G''_{-i}\rangle = |\LC_j(\LC_i(\LC_j(G)))_{-i}\rangle$ is the resultant graph state after the sequence of local complementations and vertex removal. The measurement-induced corrections for measurement outcomes $\pm$ are given as follows:
\begin{equation}
U_\pm^{(i,j)} = R_y(\mp\pi/2)_j \otimes Z^{\otimes |S_\pm|},
\label{eq:A9}
\end{equation}
where $j \in \mathcal{N}(i)$ is a selected reference neighbor, and the sets are defined as $S_+ = \mathcal{N}(i) \setminus (\mathcal{N}(j) \cup \{j\})$ and $S_- = \mathcal{N}(j) \setminus (\mathcal{N}(i) \cup \{i\})$.

Given the prominence of X-basis measurements in our analysis, we derive essential Clifford group equivalences for specific rotations. Direct computation from the matrix representations yields the following operator identities:
\begin{equation}
R_y(-\pi/2) \equiv ZH, \quad R_y(\pi/2) \equiv HZ
\label{eq:ry_equiv}
\end{equation}
with $H$ denoting the Hadamard gate. 

The state produced by a Pauli $X$ measurement is not identical to the target graph $|G''_{-i}\rangle = |\LC_j(\LC_i(\LC_j(G)))_{-i}\rangle$. However, they can be converted to each other through local operations. To return to the canonical graph state $|G''_{-i}\rangle$ after measurement, we apply a recovery operation defined by the inverse correction operator:
\begin{equation}
\big( U_\pm^{(i,j)} \big)^{-1} 
= R_y(\pm\pi/2)_j \otimes Z^{\otimes |S_\pm|}
= 
\begin{cases} 
(HZ)_j \otimes Z^{\otimes |S_+|} & \text{(} |+x\rangle_i) \\ 
(ZH)_j \otimes Z^{\otimes |S_-|} & \text{(} |-x\rangle_i)
\end{cases}
\label{eq:recovery_clifford}
\end{equation}

This follows from the basic properties $\big[ R_y(\theta) \big]^{-1} = R_y(-\theta)$ and $Z^{-1} = Z$.

It is noteworthy that studying the evolution of the corresponding pattern $|G''_{-i}\rangle = |\LC_j(\LC_i(\LC_j(G)))_{-i}\rangle$ for a specific initial pattern, a Pauli X measurement of a specific qubit, and a suitable selected reference qubit is of significant value. Figures~\ref{fig:9}--\ref{fig:11} illustrate three examples representing the pattern evolution in the three sub-protocols of our protocol. These examples demonstrate that the desired pattern can be obtained by selecting appropriate measurement and reference qubits.

\begin{figure}[htbp]
    \centering
    \includegraphics[width=\columnwidth]{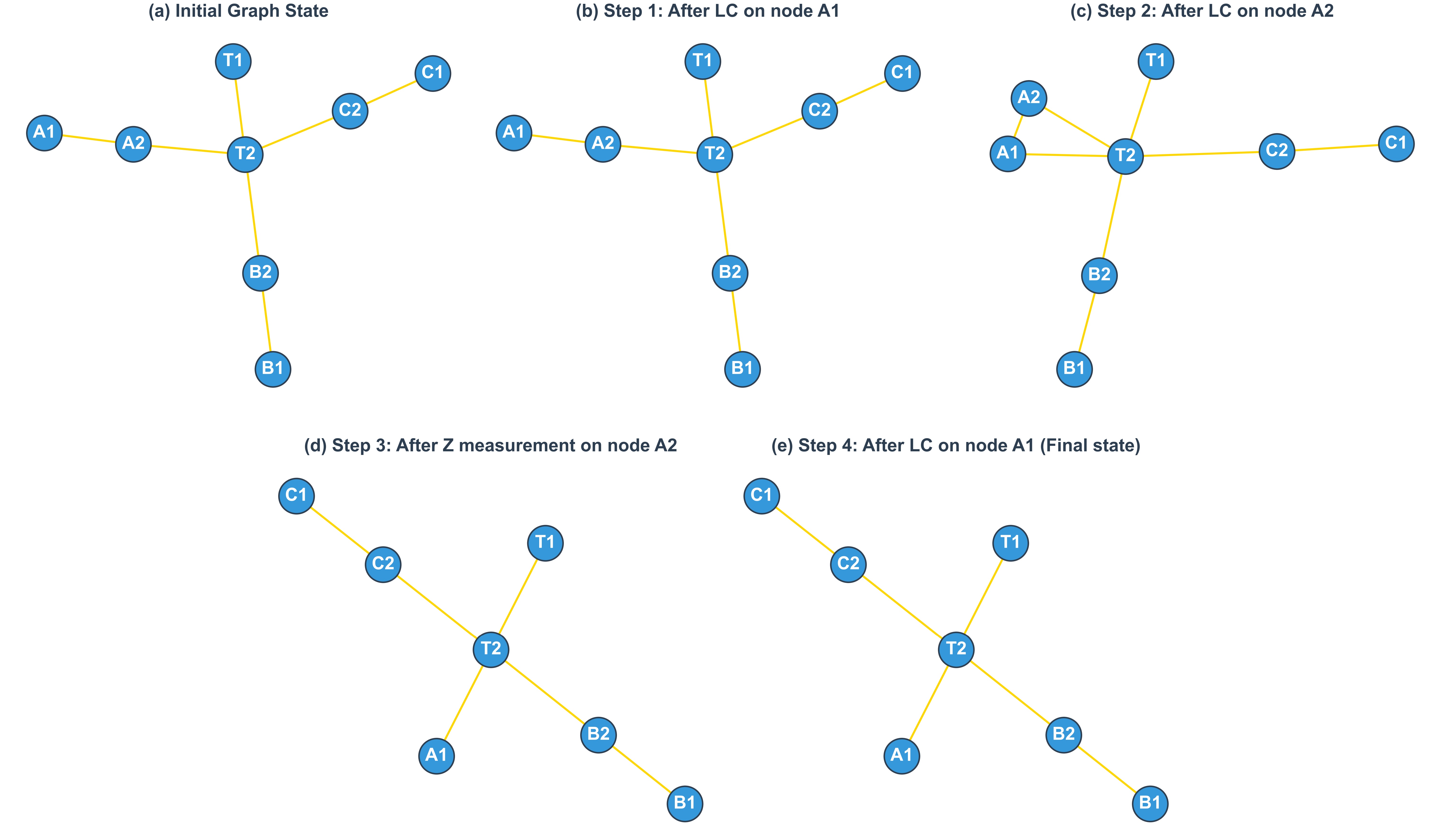}
    \caption{An $X$-basis measurement is performed on node $A_2$ using $A_1$ as the reference neighbor. The plot shows the complete evolution from the initial to the final graph state. This process extends through continuous measurements on nodes $B_2$ and $C_2$. By selecting $B_1$ and $C_1$ as the respective reference neighbors, the system ultimately yields a star graph centered at node $T_2$.}
    \label{fig:9}
\end{figure}

\begin{figure}[htbp]
    \centering
    \includegraphics[width=\columnwidth]{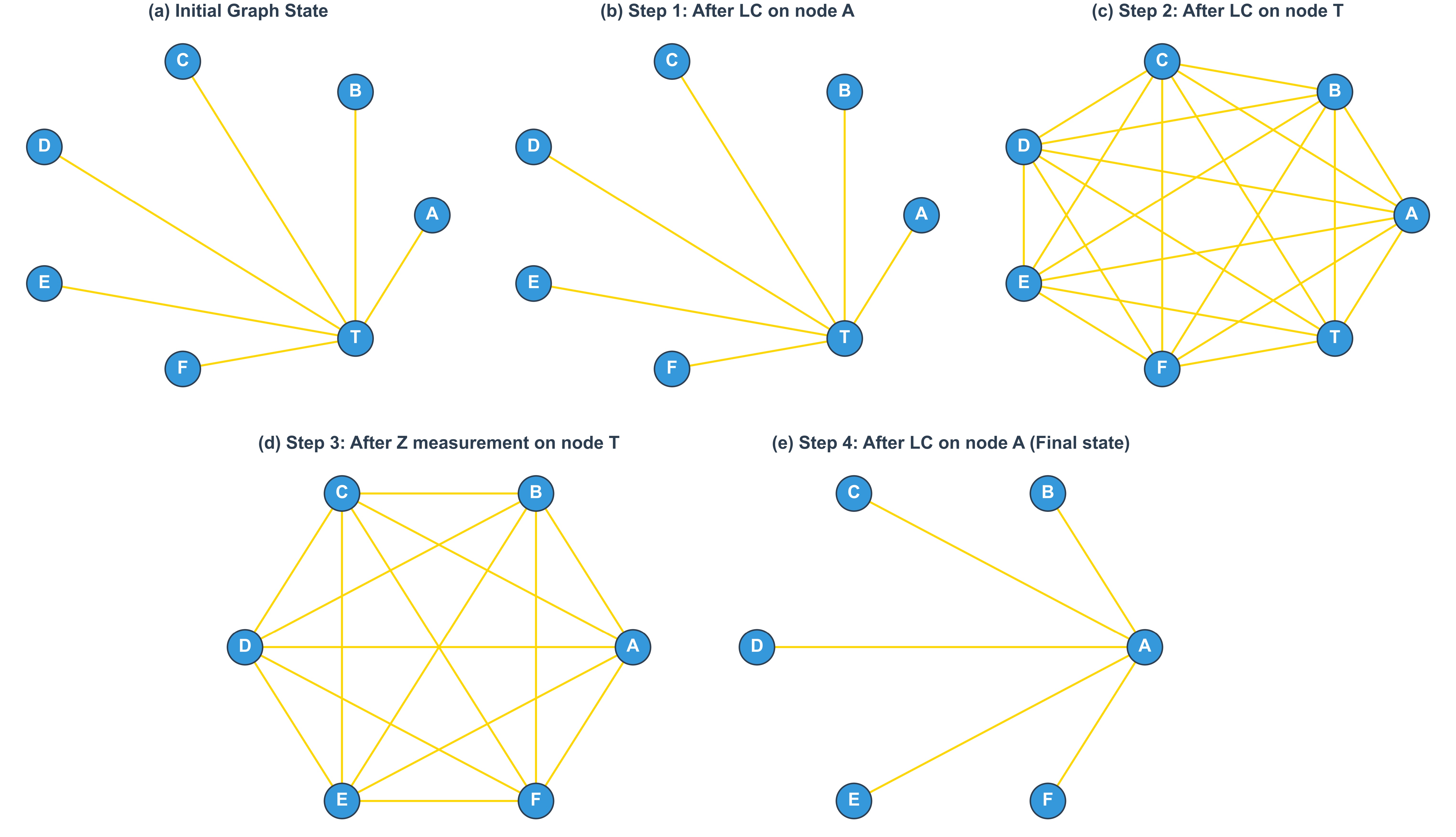}
    \caption{An $X$-basis measurement is applied to node $T$ with $A$ as the reference neighbor. This sequence highlights the evolution of the star topology during measurement operations. The final result is a new star graph centered at node $A$.}
    \label{fig:10}
\end{figure}

\begin{figure}[htbp]
    \centering
    \includegraphics[width=\columnwidth]{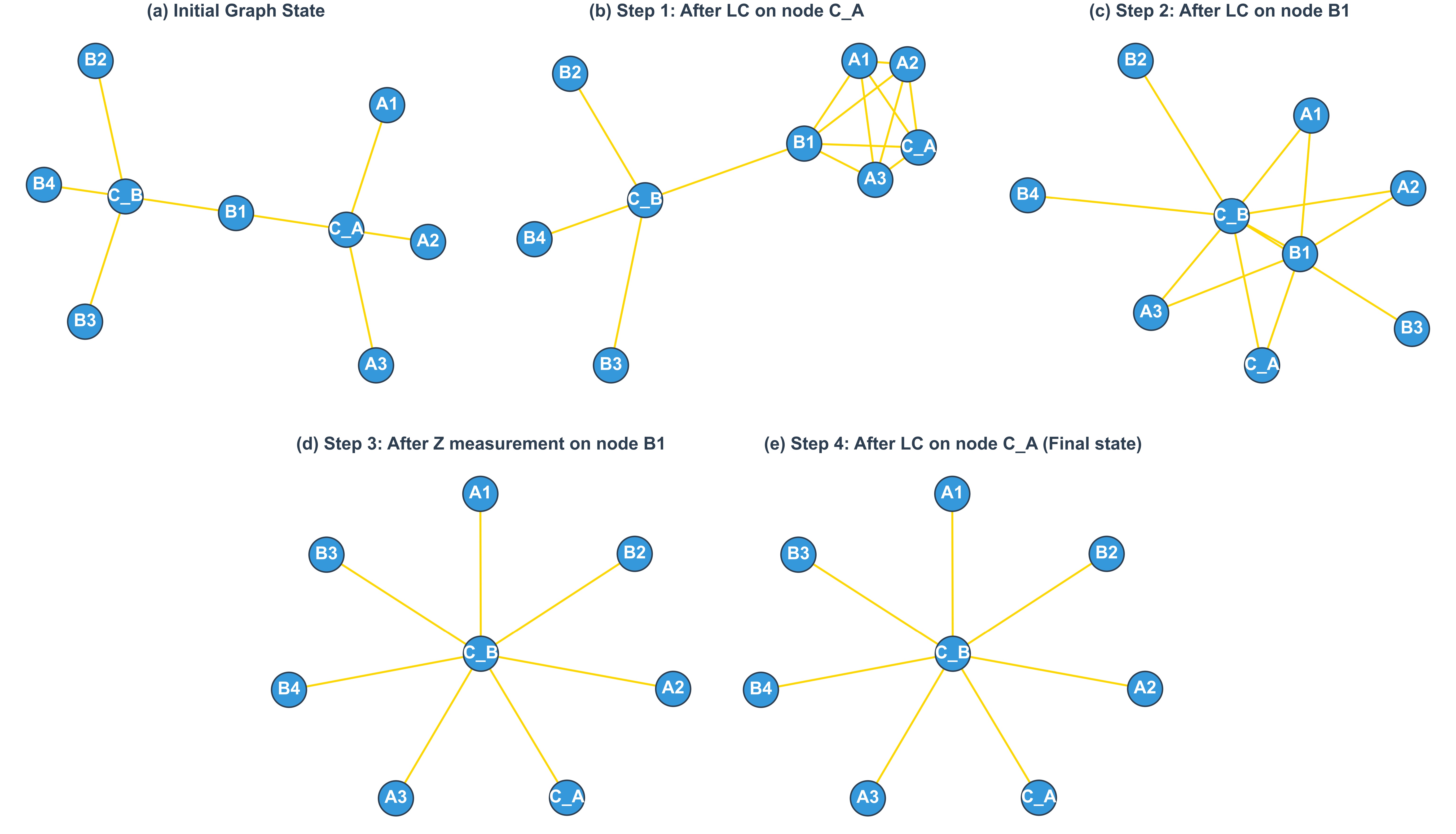}
    \caption{An $X$-basis measurement is performed on node $B_1$ with $C_A$ as the reference neighbor. This step demonstrates the transformation mechanism within complex topological structures. The measurement-induced reconstruction effectively reshapes the graph. A star graph state centered at node $C_B$ is ultimately obtained.}
    \label{fig:11}
\end{figure}

Although Eq.~\eqref{eq:recovery_clifford} indicates that the recovery process involves complex multi-qubit corrections, under certain symmetries, by strategically choosing auxiliary qubit, we can find equivalent operators that significantly reduce the number of gate operations.

Specifically, for the case where the Pauli $X$ measurement at vertex $i$ is $|+x\rangle_i$, the recovery operation exhibits significant mathematical simplification. We demonstrate this through the following logic. \textbf{Lemma 1} proves that when the neighborhood $\mathcal{N}(i)$ is an independent set, the measurement produces a local star topology centered at $j$. \textbf{Lemma 2} derives a key operator identity for the star graph: $Z_cH_c\bigotimes_{k\in L}Z_k = H_c$. Combining these conclusions, \textbf{Theorem 1} proves that, given the neighborhood $\mathcal{N}(i)$ is an independent set, the recovery operation can be simplified to a single Hadamard gate $H_j$.

For the case where the Pauli $X$ measurement result is $|-x\rangle_i$, by selecting appropriate auxiliary qubits, the qubit set that would otherwise require multi-qubit bit correction can be made empty. This also effectively simplifies the corresponding recovery operation.

\noindent \textbf{Lemma \ref{lem:star-formation} (Local Star Topology Formation).} 
Given a graph $G = (V, E)$ where the neighborhood $\mathcal{N}(i)$ of vertex $i$ is an independent set. When qubit $i$ undergoes a Pauli-$X$ measurement, the graph state evolves according to the following transformation sequence. This sequence captures the structural changes at the graph representation. For any selected neighbor $j \in \mathcal{N}(i)$, the graph topology updates as follows:
\begin{align*}
G_1 &= \mathrm{LC}_j(G), \nonumber\\
G_2 &= \mathrm{LC}_i(G_1), \nonumber\\
G_2^{-i} &= G_2 \setminus \{i\}, \nonumber\\
G'' &= \mathrm{LC}_j(G_2^{-i}).
\end{align*}
This transformation sequence generates a localized star topology in the final graph $G''$. This star graph is derived from the vertex set $\{j\} \cup S$, where the set $S = \mathcal{N}(i) \setminus \{j\}$. In this structure, $j$ is the node center, and $S$ is its surrounding leaf node.
Note that the rest of the graph $G''$ retains its original arbitrarily complex topology.

\begin{proof}
We trace the evolution of edge relationships between arbitrary distinct vertices $u, v \in S$ throughout the transformation. Initially, the independence of $\mathcal{N}(i)$ ensures $\{u,v\} \notin E(G)$. Following $\mathrm{LC}_j(G)$, the connection remains unchanged ($\{u,v\} \notin E(G_1)$) because no additional edges exist between $S$ and $j$ beyond the essential $\{i,j\}$ linkage. The subsequent $\mathrm{LC}_i$ operation introduces the missing $u$-$v$ edge ($\{u,v\} \in E(G_2)$) since both vertices neighbor $i$. This created edge persists after vertex removal ($\{u,v\} \in E(G_2^{-i})$). The concluding $\mathrm{LC}_j$ transformation then precisely removes the $u$-$v$ connection ($\{u,v\} \notin E(G'')$) because $u$ and $v$ now reside in $\mathcal{N}_j(G_2^{-i})$.

Simultaneously, the connections between $j$ and each $u \in S$ are preserved throughout all transformations. Therefore, in the subgraph of $G''$ induced by $\{j\} \cup S$, vertex $j$ is connected to all vertices in $S$, while $S$ itself remains an independent set—forming a perfect star topology.

This proof only concerns the subgraph induced by $\{j\} \cup S$. The remainder of $G''$ may have arbitrary structure unaffected by our analysis.
\end{proof}

\noindent \textbf{Lemma \ref{lem:star-phase} (Operator Reduction in Star Topology).} 
For a star graph state $|G\rangle$ with center $c$ and leaves $L$ (where $|L| \geq 1$), the following operator equivalence exists:
\begin{equation}
Z_c H_c \bigotimes_{k \in L} Z_k |G\rangle = H_c |G\rangle
\end{equation}
This identity demonstrates how multi-qubit $Z$ operations coordinated through star entanglement collapse onto a single Hadamard transformation. This equivalence holds universally for any number of leaves, including the minimal cases:
\begin{itemize}
    \item \textbf{Single leaf} ($|L|=1$): Reduces to basic two-qubit entanglement
    \item \textbf{Two leaves} ($|L|=2$): Forms a three-qubit linear chain
    \item \textbf{Multiple leaves} ($|L|\geq3$): Standard star configuration
\end{itemize}

\begin{proof}
Leveraging the stabilizer formalism, where $K_c = X_c \prod_{k \in L} Z_k$ and $K_k = Z_c X_k$ for each leaf $k \in L$, we derive the operator equivalence through quantum algebraic manipulation:
\begin{align*}
Z_c H_c \bigotimes_{k \in L} Z_k |G\rangle 
&= Z_c H_c \left( \prod_{k \in L} Z_k |G\rangle \right) \\
&= Z_c H_c (X_c |G\rangle) \quad \text{since } \textstyle\prod_{k \in L} Z_k |G\rangle = X_c |G\rangle \\
&= Z_c (Z_c H_c) |G\rangle \quad \text{since}\  H_c X_c = Z_c H_c \\
&= (Z_c^2) H_c |G\rangle \\
&= H_c |G\rangle.
\end{align*}
The equality holds up to global phase, with measurement outcomes harmonized through the star entanglement structure. 

\textit{Remark on leaf cardinality}: The proof remains valid for all $|L| \geq 1$ because:
\begin{enumerate}
    \item For $|L|=1$: The stabilizer relation $\prod_{k\in L}Z_k |G\rangle = Z_k |G\rangle = X_c |G\rangle$ holds directly
    \item For $|L|=2$: $Z_{k_1}Z_{k_2} |G\rangle = X_c |G\rangle$ follows from $K_c |G\rangle = |G\rangle$
    \item For $|L|\geq3$: The product $\prod Z_k |G\rangle$ consistently equals $X_c |G\rangle$
\end{enumerate}
This universal validity underscores the operator equivalence's robustness across all star topology scales.
\end{proof}

\noindent \textbf{Theorem \ref{thm:x-simplify} (Simplification Theorem for $|+x\rangle$ Measurement).} 
When $\mathcal{N}_i(G)$ constitutes an independent set and the $X$-basis measurement at vertex $i$ yields $|+x\rangle$, the recovery operation simplifies to:
\begin{equation}
(U_+^{(i,j)})^{-1} = H_j.
\end{equation}

\begin{proof}
Since $\mathcal{N}(i)$ is an independent set and $j \in \mathcal{N}(i)$, it follows that $j$ cannot be adjacent to any other vertex in $\mathcal{N}(i)$ in the original graph $G$. Therefore:
\[
S_+ = \mathcal{N}(i) \setminus (\mathcal{N}(j) \cup \{j\}) = \mathcal{N}(i) \setminus \{j\} = S.
\]

By Lemma~\ref{lem:star-formation}, the subgraph induced by $\{j\} \cup S$ in $|G''_{-i}\rangle$ forms a star topology.

Consider the post-measurement state:
\begin{align*}
U_+^{(i,j)} |G''\rangle &= \left( R_y(-\pi/2)_j \otimes \bigotimes_{u \in S_+} Z_u \right) |G''\rangle \\
&= \left( Z_j H_j \otimes \bigotimes_{u \in S} Z_u \right) |G''\rangle.
\end{align*}

Applying Lemma~\ref{lem:star-phase} with $c = j$ and $L = S$:
\[
Z_j H_j \bigotimes_{u \in S} Z_u |G''\rangle = H_j |G''\rangle,
\]
which implies $U_+^{(i,j)} |G''\rangle = H_j |G''\rangle$. 
Since $H_j$ is self-inverse, the recovery operation is:
\[
(U_+^{(i,j)})^{-1} = H_j.
\]
\end{proof}

\section{Proof of Protocol 1}
\label{app:protocol1}

\subsection*{Initial Resource Preparation}
The protocol initializes with the relay node $\mathcal{R}$ establishing Bell pairs with all peripheral nodes. The initial resource state is:
\begin{equation}
    \ket{\Psi_0} = \left( \bigotimes_{X \in \mathcal{P}} \ket{\Phi^+}_{X_1 X_2} \right) \otimes \ket{\Phi^+}_{T_1 T_2}, 
\end{equation}
where $\ket{\Phi^+} = \frac{1}{\sqrt{2}}(\ket{00} + \ket{11})$, $\mathcal{P} = \{A,B,C,\ldots\}$ denotes the $n$ peripheral nodes, and $T_1, T_2$ are the selected Bell pair designated as the control center.

\subsection*{Step 1: Basis Transformation}
Each Bell pair is transformed by local Hadamard gates on the relay qubits. For the peripheral links, $H_{X_2} \ket{\Phi^+} \mapsto \ket{\LC}_{X_1 X_2} \equiv \frac{1}{\sqrt{2}} (\ket{0}_{X_1}\ket{+}_{X_2} + \ket{1}_{X_1}\ket{-}_{X_2})$. Applying $H_{T_2}$ to the center link and expanding, the global state becomes:
\begin{equation}
\begin{split}
    \ket{\Psi_1} &= \frac{1}{\sqrt{2^{n+1}}} \Bigg( \ket{+}_{T_1}\ket{0}_{T_2} \\
    &\quad + \ket{-}_{T_1}\ket{1}_{T_2} \Bigg) \bigotimes_{X \in \mathcal{P}} \ket{\LC}_{X_1 X_2}.
\end{split}
\label{eq:B_psi1}
\end{equation}

\subsection*{Step 2: Entanglement Expansion}
The global Controlled-Z (CZ) operation applied between $T_2$ and all $X_2$ qubits acts conditionally based on the state of $T_2$:
\begin{align}
    \CZ_{T_2 X_2} \ket{0}_{T_2} \ket{\psi}_{X_2} &= \ket{0}_{T_2} \ket{\psi}_{X_2}, \\
    \CZ_{T_2 X_2} \ket{1}_{T_2} \ket{\pm}_{X_2} &= \ket{1}_{T_2} \ket{\mp}_{X_2}.
\end{align}
Applying $\prod_{X \in \mathcal{P}} \CZ_{T_2 X_2}$ to $\ket{\Psi_1}$, the component associated with $\ket{0}_{T_2}$ remains unchanged ($\ket{\phi_0} \equiv \sqrt{2}\ket{\LC}_{X_1 X_2}$), while the component associated with $\ket{1}_{T_2}$ undergoes a local $Z_{X_2}$ transformation ($\ket{\phi_1} \equiv \sqrt{2}(I \otimes Z)\ket{\LC}_{X_1 X_2}$). This generates the intermediate state:
\begin{equation}
\begin{split}
    \ket{\Psi_2} &= \frac{1}{\sqrt{2^{n+1}}} \Bigg( \ket{+}_{T_1}\ket{0}_{T_2} \bigotimes_{X \in \mathcal{P}} \ket{\phi_0}_{X_1 X_2} \\
    &\quad + \ket{-}_{T_1}\ket{1}_{T_2} \bigotimes_{X \in \mathcal{P}} \ket{\phi_1}_{X_1 X_2} \Bigg).
\end{split}
\end{equation}

\subsection*{Step 3: Measurement-Induced Collapse}
We next measure all $X_2$ qubits in the X-basis with outcomes $\mathbf{m} = \{m_X \in \{0,1\}\}$. The measurement projector is $\Pi_{\mathbf{m}} = \bigotimes_{X \in \mathcal{P}} \ket{m_X}\bra{m_X}_{X_2}$, where $\ket{m_X} = \ket{+}$ for $m_X=0$ and $\ket{-}$ for $m_X=1$. Evaluating the inner products for each node pair yields:
\begin{align}
    \bra{m_X}_{X_2} \ket{\phi_0}_{X_1 X_2} &= \ket{m_X}_{X_1}, \\
    \bra{m_X}_{X_2} \ket{\phi_1}_{X_1 X_2} &= \ket{m_X \oplus 1}_{X_1}.
\end{align}
Applying $\Pi_{\mathbf{m}}$ projects $\ket{\Psi_2}$ onto two orthogonal branches. Given the orthogonality, the success probability of this specific outcome configuration is exactly $p(\mathbf{m}) = \|\Pi_{\mathbf{m}} \ket{\Psi_2}\|^2 = 1/2^n$. Normalizing the state by dividing by $\sqrt{p(\mathbf{m})}$ yields the collapsed state:
\begin{equation}
\begin{split}
    \ket{\Psi_3} &= \frac{1}{\sqrt{2}} \Bigg( \ket{+}_{T_1}\ket{0}_{T_2} \bigotimes_{X} \ket{m_X}_{X_1} \\
    &\quad + \ket{-}_{T_1}\ket{1}_{T_2} \bigotimes_{X} \ket{m_X \oplus 1}_{X_1} \Bigg).
\end{split}
\end{equation}

\subsection*{Step 4: Deterministic Recovery and Final State}
To deterministically route the entanglement, we designate $X_1$ as the reference neighbor for each peripheral node $X \in \mathcal{P}$. This generates the topological sets:
\begin{align}
S_+ &= \mathcal{N}(X) \setminus (\mathcal{N}(X_1) \cup \{X_1\}) = \{T_1\}, \\
S_- &= \mathcal{N}(X_1) \setminus (\mathcal{N}(X) \cup \{X\}) = \emptyset.
\end{align}
The condition $S_- = \emptyset$ guarantees that no $Z$-corrections are needed for $|-x\rangle_X$ outcomes. Furthermore, since $\mathcal{N}(X) = \{T_1, X_1\}$ forms an independent set, the prerequisite of Theorem~\ref{thm:x-simplify} is strictly satisfied. 

By partitioning the peripheral nodes based on their measurement outcomes as $\mathcal{P}_0 = \{X \mid m_X = 0\}$ and $\mathcal{P}_1 = \{X \mid m_X = 1\}$, the protocol applies the local corrective unitary operation:
\begin{equation}
    \mathcal{U}_{\text{rec}} = \bigotimes_{X \in \mathcal{P}_0} H_{X_1} \otimes \bigotimes_{X \in \mathcal{P}_1} (Z_{X_1}H_{X_1} ).
\end{equation}

Let $\mathcal{C}_X$ represent this conditional correction on node $X_1$. It homogeneously maps the computational basis states:
\begin{align}
    \mathcal{C}_X \ket{m_X}_{X_1} &= \ket{+}_{X_1} \quad \text{for all } X \in \mathcal{P}, \\
    \mathcal{C}_X \ket{m_X \oplus 1}_{X_1} &= \ket{-}_{X_1} \quad \text{for all } X \in \mathcal{P}.
\end{align}
Applying $\mathcal{U}_{\text{rec}}$ to $\ket{\Psi_3}$ thus disentangles the measurement dependencies, homogenizing the peripheral qubit states:
\begin{equation}
\begin{split}
    \mathcal{U}_{\text{rec}} \ket{\Psi_3} &= \frac{1}{\sqrt{2}} \Bigg( \ket{+}_{T_1}\ket{0}_{T_2} \bigotimes_{X} \ket{+}_{X_1} \\
    &\quad + \ket{-}_{T_1}\ket{1}_{T_2} \bigotimes_{X} \ket{-}_{X_1} \Bigg).
\end{split}
\end{equation}
Recognizing that the neighborhood of the terminal node $T_2$ is exactly $N(T_2) = \{T_1\} \cup \{X_1 \mid X \in \mathcal{P}\}$, the state is compactly rewritten as the target star graph:
\begin{equation}
\ket{\Psi_4} = \frac{1}{\sqrt{2}} \left( \ket{0}_{T_2} \bigotimes_{v \in N(T_2)} \ket{+}_v \mkern-6mu + \mkern-6mu \ket{1}_{T_2} \bigotimes_{v \in N(T_2)} \ket{-}_v \right).
\end{equation}

\section{Proof of Protocol 2}
\label{app:protocol2}

\subsection*{Initial Resource Preparation}
The protocol initializes with a star graph state centered on node $T$:
\begin{equation}\label{eq:initial_state}
|G\rangle = \left( \prod_{v \in N(T)} \!\!\!\cz_{T v} \right) |+\rangle^{\otimes |V|},
\end{equation}
where $V = \{T\} \cup N(T)$ with $|V| = d+1$, and $N(T) = \{A,B,C,\ldots\}$ denotes $d$ peripheral nodes. The state admits the fundamental decomposition:
\begin{equation}\label{eq:decomp}
|G\rangle = \frac{1}{\sqrt{2}} \left( |0\rangle_T \bigotimes_{v \in N(T)} |+\rangle_v \mkern-6mu + \mkern-6mu |1\rangle_T \bigotimes_{v \in N(T)} |-\rangle_v \right).
\end{equation}

\subsection*{Step 1: Center Measurement and State Collapse}
Measurement of the original node $T$ in the Pauli-X basis yields outcomes $M_T \in \{0,1\}$. Decomposing the computational basis using X-basis representations ($|0\rangle_T = \frac{|+\rangle_T + |-\rangle_T}{\sqrt{2}}$ and $|1\rangle_T = \frac{|+\rangle_T - |-\rangle_T}{\sqrt{2}}$), substitution into Eq.~\eqref{eq:decomp} yields:
\begin{equation}
\begin{split}
|G\rangle &= \frac{1}{2} \Bigg[ |+\rangle_T \left( \bigotimes_{v} |+\rangle_v + \bigotimes_{v} |-\rangle_v \right) \\
&\quad + |-\rangle_T \left( \bigotimes_{v} |+\rangle_v - \bigotimes_{v} |-\rangle_v \right) \Bigg].
\end{split}
\end{equation}

The projection operators for the measurement outcomes are $\Pi_{M_T=0} = \ket{+}\bra{+}_T$ and $\Pi_{M_T=1} = \ket{-}\bra{-}_T$. Applying these projectors, both branches yield a success probability of $p(M_T=0) = p(M_T=1) = 1/2$. The normalized post-measurement states are derived as:
\begin{align}
\ket{\psi_{\text{post}}^{(0)}} &= \frac{1}{\sqrt{2}} \left( \bigotimes_{v \in N(T)} |+\rangle_v + \bigotimes_{v \in N(T)} |-\rangle_v \right), \quad M_T = 0, \\
\ket{\psi_{\text{post}}^{(1)}} &= \frac{1}{\sqrt{2}} \left( \bigotimes_{v \in N(T)} |+\rangle_v - \bigotimes_{v \in N(T)} |-\rangle_v \right), \quad M_T = 1.
\end{align}

\subsection*{Step 2: Topology Reconfiguration via Deterministic Recovery}
Here, the recovery operation has the same form as the recovery operation in Protocol 1. When designating the new center $A$ as the reference neighbor, the topological sets exhibit critical degeneracies:
\begin{align}
S_+ &= \mathcal{N}(T) \setminus (\mathcal{N}(A) \cup \{A\}) = \{B,C,\ldots\}, \\
S_- &= \mathcal{N}(A) \setminus (\mathcal{N}(T) \cup \{T\}) = \emptyset.
\end{align}

These set properties exhibit two fundamental simplification characteristics. First, $S_- = \emptyset$ signifies that for $|-x\rangle_A$ measurement outcomes, no $Z$-corrections are required on the corresponding neighborhood set, as the correction domain degenerates to the null set. Second, the set $\mathcal{N}(T) = \{A,B,C,\ldots\}$ clearly satisfies the independence condition of Theorem~\ref{thm:x-simplify} due to the bipartite connectivity structure of a star topology. Consequently, the topological constraints $S_- = \emptyset$ and the bipartite isolation of $S_+$ reduce the recovery operation directly to a local unitary on $A$:
\begin{equation}
\mathcal{U}_{\text{rec}} = 
\begin{cases} 
    H_A & M_T = 0, \\
    Z_A H_A & M_T = 1. 
\end{cases}
\end{equation}

The correction procedure is analyzed for both measurement outcomes. For the case $M_T=0$, applying the Hadamard gate to node $A$ transforms the post-measurement state:
\begin{equation}
H_A \ket{\psi_{\text{post}}^{(0)}} = \frac{1}{\sqrt{2}} \left( |0\rangle_A \bigotimes_{u \neq A} |+\rangle_u + |1\rangle_A \bigotimes_{u \neq A} |-\rangle_u \right).
\end{equation}

For the case $M_T=1$, the sequential application of Hadamard and phase gates yields:
\begin{equation}
H_A \ket{\psi_{\text{post}}^{(1)}} = \frac{1}{\sqrt{2}} \left( |0\rangle_A \bigotimes_{u \neq A} |+\rangle_u - |1\rangle_A \bigotimes_{u \neq A} |-\rangle_u \right),
\end{equation}
and applying $Z_A$ flips the relative phase to recover the target configuration:
\begin{equation}
Z_A H_A \ket{\psi_{\text{post}}^{(1)}} = \frac{1}{\sqrt{2}} \left( |0\rangle_A \bigotimes_{u \neq A} |+\rangle_u + |1\rangle_A \bigotimes_{u \neq A} |-\rangle_u \right).
\end{equation}

Both cases converge to the identical canonical form after correction:
\begin{equation}
|G_{\text{final}}\rangle = \frac{1}{\sqrt{2}} \left( |0\rangle_A \bigotimes_{u \in N'(A)} |+\rangle_u + |1\rangle_A \bigotimes_{u \in N'(A)} |-\rangle_u \right),
\end{equation}
which is exactly the star graph state $\left( \prod_{u \in N'(A)} \cz_{A u} \right) |+\rangle^{\otimes |V|-1}$ centered at node $A$, establishing a new star topology with preserved peripheral connectivity $N'(A) = N(T) \setminus \{A\}$.

\section{Proof of Protocol 3}
\label{app:protocol3}

\subsection*{Initial Resource Preparation}
The protocol begins with two independent star-shaped graph states:
\begin{align}
|G_A\rangle &= \frac{1}{\sqrt{2}} \Bigg( |0\rangle_{C_A} \bigotimes_{u \in N_A} |+\rangle_u + |1\rangle_{C_A} \bigotimes_{u \in N_A} |-\rangle_u \Bigg), \\
|G_B\rangle &= \frac{1}{\sqrt{2}} \Bigg( |0\rangle_{C_B} \bigotimes_{v \in N_B} |+\rangle_v + |1\rangle_{C_B} \bigotimes_{v \in N_B} |-\rangle_v \Bigg),
\end{align}
where $C_A$ and $C_B$ are the central nodes, $N_A$ and $N_B$ are peripheral nodes, and $B_1 \in N_B$ is designated as the connection point. The combined initial state is $|\Psi_0\rangle = |G_A\rangle \otimes |G_B\rangle$.

\subsection*{Step 1: Entanglement Link Generation}
Applying the controlled-Z (CZ) operator between $C_A$ and $B_1$ acts as a conditional phase flip. When $C_A = |0\rangle$, $B_1$ remains unchanged; when $C_A = |1\rangle$, the basis states $|+\rangle_{B_1}$ and $|-\rangle_{B_1}$ in $G_B$ are flipped to $|-\rangle_{B_1}$ and $|+\rangle_{B_1}$, respectively. The CZ operation transforms the state as:
\begin{equation}
\begin{split}
|\Psi_1\rangle &= \frac{1}{\sqrt{2}} \Bigg( |0\rangle_{C_A} \bigotimes_{u \in N_A} |+\rangle_u \otimes |G_B\rangle \\
&\quad + |1\rangle_{C_A} \bigotimes_{u \in N_A} |-\rangle_u \otimes |\widetilde{G_B}\rangle \Bigg),
\end{split}
\end{equation}

where the transformed $G_B$ component is:

\begin{equation}
|\widetilde{G_B}\rangle = |0\rangle_{C_B} |-\rangle_{B_1} \bigotimes_{v \neq B_1} |+\rangle_v + |1\rangle_{C_B} |+\rangle_{B_1} \bigotimes_{v \neq B_1} |-\rangle_v.
\end{equation}

\subsection*{Step 2: Measurement-Induced Graph Fusion}
We measure node $B_1$ in the Pauli-X basis. The projection onto $M_{B_1}=0$ ($\langle +|_{B_1}$) evaluates the inner products with the respective $G_B$ and $\widetilde{G_B}$ components:
\begin{equation}
\begin{split}
\langle +|_{B_1} |\Psi_1\rangle &= \frac{1}{2} \Bigg( |0\rangle_{C_A} \bigotimes_{u \in N_A} |+\rangle_u \langle +|_{B_1} |G_B\rangle \\
&\quad + |1\rangle_{C_A} \bigotimes_{u \in N_A} |-\rangle_u \langle +|_{B_1} |\widetilde{G_B}\rangle \Bigg) \\
&= \frac{1}{2} \Bigg( |0\rangle_{C_A} \bigotimes_{u \in N_A} |+\rangle_u |0\rangle_{C_B} \bigotimes_{v \neq B_1} |+\rangle_v \\
&\quad + |1\rangle_{C_A} \bigotimes_{u \in N_A} |-\rangle_u |1\rangle_{C_B} \bigotimes_{v \neq B_1} |-\rangle_v \Bigg).
\end{split}
\end{equation}
Similarly, projection onto $M_{B_1}=1$ ($\langle -|_{B_1}$) yields:
\begin{equation}
\begin{split}
\langle -|_{B_1} |\Psi_1\rangle &= \frac{1}{2} \Bigg( |0\rangle_{C_A} \bigotimes_{u \in N_A} |+\rangle_u |1\rangle_{C_B} \bigotimes_{v \neq B_1} |-\rangle_v \\
&\quad + |1\rangle_{C_A} \bigotimes_{u \in N_A} |-\rangle_u |0\rangle_{C_B} \bigotimes_{v \neq B_1} |+\rangle_v \Bigg).
\end{split}
\end{equation}
For both outcomes, the normalization factor is derived from the squared norms $\left\| \langle \pm|_{B_1} |\Psi_1\rangle \right\|^2 = 1/4 + 1/4 = 1/2$. Multiplying by $\sqrt{2}$ yields the normalized post-measurement states $|\Psi^{(0)}\rangle$ and $|\Psi^{(1)}\rangle$.

\subsection*{Step 3: Topology Regularization via Local Recovery}
The recovery operation now targets node $C_A$ as the reference neighbor. The topological sets are computed relative to the measured node $B_1$ and reference neighbor $C_A$:
\begin{align}
S_+ &= \mathcal{N}(B_1) \setminus (\mathcal{N}(C_A) \cup \{C_A\}) = \{C_B\}, \\
S_- &= \mathcal{N}(C_A) \setminus (\mathcal{N}(B_1) \cup \{B_1\}) = N_A.
\end{align}
Critical topological properties emerge from the computed sets: The neighborhood $\mathcal{N}(B_1) = \{C_A, C_B\}$ constitutes an independent set, satisfying the independence prerequisite of Theorem~\ref{thm:x-simplify}. Concurrently, $S_- = N_A$ necessitates $Z$-corrections on $C_A$'s peripheral nodes for $|-x\rangle_{B_1}$ measurement outcomes. 

Application of Theorem~\ref{thm:x-simplify} yields the minimal recovery operator:
\begin{equation}
\mathcal{U}_{\text{rec}} = 
\begin{cases} 
    H_{C_A} & M_{B_1} = 0, \\
    \left( \prod_{u \in N_A} Z_u \right) (Z_{C_A} H_{C_A}) & M_{B_1} = 1.
\end{cases}
\end{equation}

For the case $M_{B_1}=0$, applying $H_{C_A}$ directly transforms the core qubits:
\begin{equation}
\begin{split}
H_{C_A} |\Psi^{(0)}\rangle &= \frac{1}{\sqrt{2}} \Bigg( |+\rangle_{C_A} \bigotimes_{u \in N_A} |+\rangle_u |0\rangle_{C_B} \bigotimes_{v \neq B_1} |+\rangle_v \\
&\quad + |-\rangle_{C_A} \bigotimes_{u \in N_A} |-\rangle_u |1\rangle_{C_B} \bigotimes_{v \neq B_1} |-\rangle_v \Bigg).
\end{split}
\end{equation}

For the case $M_{B_1}=1$, sequential application of $H_{C_A}$ and $Z_{C_A}$ yields:
\begin{equation}
\begin{split}
Z_{C_A} H_{C_A} |\Psi^{(1)}\rangle &= \frac{1}{\sqrt{2}} \Bigg( |-\rangle_{C_A} \bigotimes_{u \in N_A} |+\rangle_u |1\rangle_{C_B} \bigotimes_{v \neq B_1} |-\rangle_v \\
&\quad + |+\rangle_{C_A} \bigotimes_{u \in N_A} |-\rangle_u |0\rangle_{C_B} \bigotimes_{v \neq B_1} |+\rangle_v \Bigg).
\end{split}
\end{equation}
Applying the $Z$ gates to all peripheral nodes in $N_A$ flips the relative phases of the $u$ registers, rectifying the state to:
\begin{equation}
\begin{split}
&\left( \prod_{u \in N_A} Z_u \right) Z_{C_A} H_{C_A} |\Psi^{(1)}\rangle \\
&= \frac{1}{\sqrt{2}} \Bigg( |-\rangle_{C_A} \bigotimes_{u \in N_A} |-\rangle_u |1\rangle_{C_B} \bigotimes_{v \neq B_1} |-\rangle_v \\
&\quad + |+\rangle_{C_A} \bigotimes_{u \in N_A} |+\rangle_u |0\rangle_{C_B} \bigotimes_{v \neq B_1} |+\rangle_v \Bigg).
\end{split}
\end{equation}

Both cases converge identically. Rewriting the final outcome in the standard star graph format centered at $C_B$, we obtain:
\begin{equation}
|G_{\text{fused}}\rangle = \frac{1}{\sqrt{2}} \left( |0\rangle_{C_B} \bigotimes_{w \in N_{\text{fused}}} |+\rangle_w + |1\rangle_{C_B} \bigotimes_{w \in N_{\text{fused}}} |-\rangle_w \right),
\end{equation}

\end{document}